\documentclass[format=manuscript, 11pt, review=false, screen=true, nonacm=true]{acmart}
\usepackage{booktabs} 

\usepackage{amsmath}
\DeclareMathAlphabet{\mathsf}{OT1}{LibertinusSans-LF}{m}{n}
\SetMathAlphabet{\mathsf}{bold}{OT1}{LibertinusSans-LF}{bx}{n}

\usepackage{import}

\graphicspath{ {./img/} }

\usepackage{standalone}
\usepackage{tikz}
\usetikzlibrary {positioning,shapes,calc,arrows.meta,graphs,patterns,backgrounds,decorations,decorations.pathreplacing}

\usepackage{subcaption}
\usepackage[colorinlistoftodos, prependcaption, textsize=scriptsize]{todonotes}

\usepackage[ruled,vlined,linesnumbered]{algorithm2e}
\DontPrintSemicolon
\SetKwInOut{Input}{\footnotesize Input}
\SetKwInOut{Output}{\footnotesize Output}
\SetKw{KwDownTo}{downto}
\SetKw{KwAnd}{and}
\SetKw{KwOr}{or}
\SetKw{KwNot}{not}
\SetKw{KwTrue}{true}
\SetKw{KwFalse}{false}
\SetKw{KwDownto}{downto}
\SetKw{Error}{error}
\SetKw{Continue}{continue}
\SetKwProg{Function}{function}{}{}
\newcommand{\IO}[2]{\Input{\footnotesize #1} \Output{\footnotesize #2}\BlankLine\BlankLine }

\SetKwProg{Fn}{Proc}{}{}

\usepackage{multicol}

\usepackage{etoolbox}
\newtoggle{rindex}
\togglefalse{rindex}
\newtoggle{srindex}
\toggletrue{srindex}

\acmJournal{TALG}
\acmVolume{9}
\acmNumber{4}
\acmArticle{39}
\acmYear{2010}
\acmMonth{3}
\copyrightyear{2009}

\setcopyright{acmlicensed}

\acmDOI{0000001.0000001}

\newcommand\Oh{{\mathcal{O}}}
\newcommand\oh{o}

\newcommand\vsAlph{\sigma}

\newcommand\vTColl{\mathcal{T}}	\newcommand\vT{\vTColl}
\newcommand\vsTColl{\mathit{n}}	\newcommand\vn{\vsTColl}
\newcommand\vPat{\mathit{P}}	\newcommand\vP{\vPat}
\newcommand\vsPat{\mathit{m}}	\newcommand\vm{\vsPat}

\newcommand\vnOcc{\mathit{occ}} \newcommand\occ{\vnOcc}

\newcommand\vSPosPat{\mathit{sp}}     \newcommand\vsp{\vSPosPat}
\newcommand\vEPosPat{\mathit{ep}}     \newcommand\vep{\vEPosPat}

\newcommand\vnBWTRuns{\mathit{r}}	\newcommand\vr{\vnBWTRuns}

\newcommand\vSamp{\mathit{t}'}
\newcommand\vt{\mathit{t}}
\newcommand\vSampFac{s}			\newcommand\vs{\vSampFac}

\newcommand\vSPosMRun{\mathit{sm}}     \newcommand\vsm{\vSPosMRun}
\newcommand\vEPosMRun{\mathit{em}}     \newcommand\vem{\vEPosMRun}

\newcommand\SA{\mathsf{SA}}
\newcommand\ISA{\mathsf{\SA\textsuperscript{-1}}}

\newcommand\CSA{\textsf{csa}}
\newcommand\FMIdx{\textsf{fm-index}}
\newcommand\RLCSA{\textsf{rlcsa}}
\newcommand\RLFMIdx{\textsf{rlfm-index}}
\newcommand\RIdx{\textsf{\textit{r}-index}}
\newcommand\RCSA{\textsf{\textit{r}-\CSA}}
\newcommand\SRIdx{\textsf{\textit{sr}-index}}
\newcommand\SRCSA{\textsf{\textit{sr}-\CSA}}
\newcommand\LZIdx{\textsf{lz-index}}
\newcommand\GIdx{\textsf{\textit{g}-index}}
\newcommand\HybridIdx{\textsf{hyb-index}}
\newcommand\LZEndIdx{\textsf{lze-index}}

\newcommand\BWT{\mathsf{BWT}}

\newcommand\LF{\mathsf{LF}}

\newcommand\fnPhi{\mathsf{\Phi}}
\newcommand\fnIPhi{\mathsf{\fnPhi\textsuperscript{-1}}}

\DeclareMathOperator{\ftSearch}{t_{search}(\vsPat)}
\DeclareMathOperator{\ftLookup}{t_{lookup}(\vsTColl)}

\newcommand\nameDS{\texttt}
\newcommand\nameOp{\texttt}
\newcommand\nameIdx{\textsf}
\newcommand\nameColl{\textsf}

\newcommand\rank{\nameOp{rank}}
\newcommand\select{\nameOp{select}}
\newcommand\predecessor{\nameOp{pred}}
\newcommand\successor{\nameOp{succ}}
\newcommand\mapping{\nameOp{map}}

\newcommand\findStartToehold{\nameOp{findStartToehold}}

\newcommand\Start{\nameDS{Start}}
\newcommand\Letter{\nameDS{Letter}}
\newcommand\First{\nameDS{First}}

\newcommand\FirstToRun{\nameDS{FirstToRun}}
\newcommand\Samples{\nameDS{Samples}}

\newcommand\SamplesF{\mathtt{F_{\SA}}}
\newcommand\MarksL{\mathtt{L_{\vTColl}}}
\newcommand\MapMarksL{\mathtt{Map_{LF}}}
\newcommand\SamplesL{\mathtt{L_{\SA}}}
\newcommand\MarksF{\mathtt{F_{\vTColl}}}
\newcommand\MapMarksF{\mathtt{Map_{FL}}}
\newcommand\Valid{\nameDS{Valid}}
\newcommand\ValidF{\mathtt{Valid_{F}}}
\newcommand\ValidL{\mathtt{Valid_{L}}}

\newcommand\ValidAreaF{\mathtt{ValidArea_{F}}}
\newcommand\ValidAreaL{\mathtt{ValidArea_{L}}}
\newcommand\Removed{\mathtt{Del}}

\begin{document}

\title{Fast and Small Subsampled R-indexes}


\author{Dustin Cobas}
\affiliation{
  \institution{CeBiB --- Center for Biotechnology and Bioengineering}
  \country{Chile}
}
\affiliation{
  \institution{University of Chile}
  \department{Dept. of Computer Science}
  \country{Chile}
}
\email{dustin.cobas@gmail.com}
\orcid{https://orcid.org/0000-0001-6081-694X}
%

\author{Travis Gagie}
\affiliation{
  \institution{CeBiB --- Center for Biotechnology and Bioengineering}
  \country{Chile}
}
\affiliation{
  \institution{Dalhousie University}
  \country{Canada}
}
\email{travis.gagie@gmail.com}
\orcid{https://orcid.org/0000-0003-3689-327X}

\author{Gonzalo Navarro}
\affiliation{
  \institution{CeBiB --- Center for Biotechnology and Bioengineering}
  \country{Chile}
}
\affiliation{
  \institution{University of Chile}
  \department{Dept. of Computer Science}
  \country{Chile}
}
\email{gnavarro@dcc.uchile.cl}
\orcid{https://orcid.org/0000-0002-2286-741X}
\begin{abstract}
  The $\RIdx$ (Gagie et al., JACM 2020) represented a breakthrough in compressed indexing of repetitive text collections, outperforming its alternatives by orders of magnitude in query time.
  Its space usage, $\Oh(r)$ where $r$ is the number of runs in the Burrows--Wheeler Transform of the text, is however higher than Lempel--Ziv and grammar-based indexes, and makes it uninteresting in various real-life scenarios of milder repetitiveness.
  In this paper we introduce the $\SRIdx$, a variant that limits a large fraction of the space to $\Oh(\min(r,n/s))$ for a text of length $n$ and a given parameter $s$, at the expense of multiplying by $s$ the time per occurrence reported.
  The $\SRIdx$ is obtained by carefully subsampling the text positions indexed by the $\RIdx$, in a way that we prove is still able to support pattern matching with guaranteed performance.
  Our experiments demonstrate that the theoretical analysis falls short in describing the practical advantages of the $\SRIdx$, because it performs much better on real texts than on synthetic ones: the $\SRIdx$ retains the performance of the $\RIdx$ while using 1.5--4.0 times less space, sharply outperforming {\em virtually every other} compressed index on repetitive texts in both time and space.
  Only a particular Lempel--Ziv-based index uses less space---about half---than the $\SRIdx$, but it is an order of magnitude slower.

  Our second contribution are the $\RCSA$ and $\SRCSA$ indexes. Just like the $\RIdx$ adapts the well-known FM-Index to repetitive texts, the $\RCSA$ adapts Sadakane's Compressed Suffix Array (CSA) to this case. We show that the principles used on the $\RIdx$ turn out to fit naturally and efficiently in the CSA framework. The $\SRCSA$ is the corresponding subsampled version of the $\RCSA$. While the CSA performs better than the FM-Index on classic texts with alphabets larger than DNA, our experiments show that the $\SRCSA$ outperforms the $\SRIdx$ on repetitive texts not only over those larger alphabets, but on some DNA texts as well.

  Overall, our new subsampled indexes sweep the table of the existing indexes for highly repetitive text collection, by combining the exceptional speed of the $\RIdx$ with drastically reduced storage use.
\end{abstract}

\maketitle

\section{Introduction} \label{sec:intro}

The rapid surge of massive repetitive text collections, like genome and sequence read sets and versioned document and software repositories, has raised the interest in text indexing techniques that exploit repetitiveness to obtain orders-of-magnitude space reductions, while supporting pattern matching directly on the compressed text representations~\cite{GNencyc18,Nav20}.

Traditional compressed indexes rely on statistical compression~\cite{NavarroM07:CFT}, but this is ineffective to capture repetitiveness~\cite{KreftN13:CIR}.
A new wave of repetitiveness-aware indexes~\cite{Nav20} build on other compression mechanisms like Lempel--Ziv~\cite{LZ76} or grammar compression~\cite{KY00}.
A particularly useful index of this kind is the $\RLFMIdx$~\cite{MakinenN05:SSA,MakinenNSV10:SRH}, because it emulates the classical suffix array~\cite{ManberM93:SAN} and this simplifies translating suffix-array based algorithms to run on it~\cite{MBCT15}.

The $\RLFMIdx$ represents the Burrows--Wheeler Transform ($\BWT$)~\cite{BurrowsW94:BSL} of the text in run-length compressed form, because the number $r$ of maximal equal-letter runs in the $\BWT$ is known to be small on repetitive texts \cite{KK19}.
A problem with the $\RLFMIdx$ is that, although it can count the number of occurrences of a pattern using $\Oh(\vr)$ space, it needs to sample the text at every $\vs$th position, for a parameter $\vs$, in order to locate each of those occurrences in time proportional to $\vs$.
The $\Oh(\vn/\vs)$ additional space incurred on a text of length $\vn$ ruins the compression on very repetitive collections, where $\vr \ll \vn$.
The recent $\RIdx$~\cite{GagieNP20:FFS} closed the long-standing problem of efficiently locating the occurrences within $\Oh(\vnBWTRuns)$ space, offering pattern matching time orders of magnitude faster than previous repetitiveness-aware indexes.

In terms of space, however, the $\RIdx$ is considerably larger than Lempel--Ziv based indexes of size $\Oh(z)$, where $z$ is the number of phrases in the Lempel--Ziv parse.
Gagie et al.~\cite{GagieNP20:FFS} show that, on extremely repetitive text collections where $n/r = 500$--$10{,}000$, $r$ is around $3z$ and the $\RIdx$ size is $0.06$--$0.2$ bits per symbol (bps), about twice that of the $\LZIdx$~\cite{KreftN13:CIR}, a baseline Lempel--Ziv index.
However, $r$ degrades faster than $z$ as repetitiveness drops: in an experiment on bacterial genomes in the same article, where $n/r \approx 100$, the $\RIdx$ space approaches $0.9$ bps, $4$ times that of the $\LZIdx$; $r$ also approaches $4z$.
Experiments on other datasets show that the $\RIdx$ tends to be considerably larger~\cite{NS19,CNP21,DN21,BCGHMNR21}. 
Indeed, while in some realistic cases $n/r$ can be over 1{,}500, in most cases it is well below: 40--160 on versioned software and document collections and fully assembled human chromosomes, 7.5--50 on virus and bacterial genomes (with $r$ in the range $4z$--$7z$), and just 4--9 on sequencing reads; see Section~\ref{sec:result}.
An $\RIdx$ on such a small $n/r$ ratio easily becomes larger than the plain sequence data.

In this paper we tackle the problem of the (relatively) large space usage of the $\RIdx$.
This index manages to locate the pattern occurrences by sampling $\vr$ text positions (corresponding to the ends of $\BWT$ runs).
We show that one can remove some carefully chosen samples so that, given a parameter $s$, the index stores only $\Oh(\min(r,n/s))$ samples while its locating machinery can still be used to guarantee that every pattern occurrence is located within $\Oh(s)$ steps.
We call the resulting index the {\em subsampled $\RIdx$}, or $\SRIdx$.
The worst-case time to locate the $\occ$ occurrences of a pattern of length $\vm$ on an alphabet of size $\sigma$ then rises from $\Oh((\vm+\occ)\log(\sigma+n/r))$ in the implemented $\RIdx$ to $\Oh((\vm+\vs \cdot \occ)\log(\sigma+n/r))$ in the $\SRIdx$, which matches the search cost of the $\RLFMIdx$.

The $\SRIdx$ can then be seen as a hybrid between the $\RIdx$ (matching it when $s=1$) and the $\RLFMIdx$ (obtaining its time with less space; the spaces become similar when repetitiveness drops).
In practice, however, the $\SRIdx$ performs {\em much better than both} on repetitive texts, retaining the time performance of the $\RIdx$ while using $1.5$--$4.0$ times less space, and sharply dominating the $\RLFMIdx$, the best grammar-based index~\cite{CNP21}, and the $\LZIdx$, both in space and time.
Its only remaining competitor is a hybrid between a Lempel--Ziv based and a statistical index~\cite{FKP18}.
This index can use up to half the space of the $\SRIdx$, but it is an order of magnitude slower.
Overall, the $\SRIdx$ stays {\em orders of magnitude faster than all the alternatives} while using small space---generally less---in a wide range of repetitiveness scenarios.

For historical reasons, the $\RIdx$ was developed on top of the $\RLFMIdx$, which performs best on small alphabets like DNA. Another well-known alternative to the $\RLFMIdx$, the $\RLCSA$  \cite{MakinenN05:SSA,MakinenNSV10:SRH}, performs better on larger alphabets but suffers from the same space-time tradeoff: one needs to spend $\Oh(n/s)$ space in order to report each occurrence in time proportional to $s$. Our second contribution is to adapt the $\RIdx$ mechanisms to run on the $\RLCSA$, to obtain what we dub {\em $\RCSA$}. It turns out that the techniques used on the $\RIdx$ apply naturally and efficiently to the $\RLCSA$ data structures, leading to a space- and time-efficient index. We further apply the subsampling mechanism of the $\SRIdx$ to the $\RCSA$ to obtain the {\em subsampled $\RCSA$}, or $\SRCSA$. Our experiments show that the $\SRCSA$ outperforms the $\SRIdx$ on texts over large alphabets, as well as on some repetitive DNA collections.

Overall, the development of the $\SRIdx$ and the $\SRCSA$ represents a major improvement in the state of the art of compressed indexes for highly repetitive text collections. By combining the speed of the $\RIdx$, which was by far the fastest index but used significantly more space than others, with a drastic reduction in space that does not sacrifice time, our new subsampled indexes sharply dominate most of the existing actors in compressed text indexing.

A conference version of this paper appeared in {\em Proc. CPM 2021} \cite{CGN21}. This article contains more detailed explanations, more extensive experiments, improved implementations, and the full development of the $\RCSA$ and $\SRCSA$ indexes.

\section{Background} \label{sec:background}


\subsection{Suffix arrays}

The {\em suffix array}~\cite{ManberM93:SAN} $\SA[1..\vsTColl]$ of a string $\vTColl[1..\vsTColl]$ over alphabet $[1..\sigma]$ is a permutation of the starting positions of all the suffixes of $\vTColl$ in lexicographic order, $\vTColl[\SA[i].. \vsTColl] < \vTColl[\SA[i + 1].. \vsTColl]$ for all $1 \leq i < \vsTColl$.
For technical convenience we assume that $\vTColl[\vsTColl]=\$$, a special terminator symbol that is smaller than every other symbol in $\vTColl$.
The suffix array can be binary searched in time $\Oh(\vsPat \log \vsTColl)$ to obtain the range $\SA[\vSPosPat.. \vEPosPat]$ of all the suffixes prefixed by a search pattern $\vPat[1..\vsPat]$. Once this range is determined, the occurrences of $\vPat$ can be {\em counted} (i.e., return their number $\occ = \vEPosPat - \vSPosPat + 1$ of occurrences in $\vTColl$), and also {\em located} (i.e., returning their positions in $\vTColl$) in time $\Oh(\occ)$ by simply listing their starting positions, $\SA[\vSPosPat],\ldots, \SA[\vEPosPat]$.
The suffix array can then be stored in $\vsTColl \lceil \log \vsTColl \rceil$ bits\footnote{We use $\log$ to denote the binary logarithm.} (plus the $n\lceil\log \sigma\rceil$ bits to store $\vT$) and we say it {\em searches} (i.e., counts and locates) for $\vP$ in $\vT$ in total time $\Oh(m\log n + \occ)$. Its main drawback is that its space usage is too high to maintain it in main memory for current text collections.

\subsection{Compressed suffix arrays} \label{subsec:csa}

\emph{Compressed suffix arrays (CSAs)}~\cite{NavarroM07:CFT} are space-efficient representations of both the suffix array ($\SA$) and the text ($\vT$).
They can find the interval $\SA[\vSPosPat..\vEPosPat]$ corresponding to $\vPat[1..\vsPat]$ in time $\ftSearch$ and access any cell $\SA[i]$ in time $\ftLookup$, so they can be used to search for $\vPat$ in time $\Oh(\ftSearch + \occ \ftLookup)$.

\subsubsection{$\Psi$--based CSAs} Grossi and Vitter \cite{GrossiV05:CSA} and Sadakane \cite{Sad03} introduced CSAs based on another permutation, $\Psi[1..n]$, related to the suffix array:
$$ \Psi(i) = \SA^{-1}[(\SA[i] \bmod \vsTColl) + 1],$$
that is, $\Psi(i) = j$ such that $\SA[\Psi(i)] = \SA[j] = \SA[i] + 1$.
$\Psi$ is then a permutation of $[1..\vsTColl]$ where $\Psi(i)$ holds the position of $\SA[i]+1$ in $\SA$, which allows us virtually move forward in $\vTColl$ from $\SA$: if $\SA[i]$ points to $\vTColl[j]$, then $\SA[\Psi(i)]$ points to $\vTColl[j+1]$. Sadakane's CSA \cite{Sad03}, which we call simply $\CSA$, adds a bitvector $D[1..n]$ that marks with $D[i]=1$ the $\sigma$ positions $\SA[i]$ where the first symbol of the suffixes changes, and an $o(n)$-space data structure \cite{Cla96} that computes in constant time $\rank(D,i)$, the number of 1s in $D[1..i]$. This allows us reading any string pointed from $\SA[j]$: the consecutive symbols are $\rank(D,i)$, $\rank(D,\Psi(i))$, $\rank(D,\Psi(\Psi(i)))$, $\ldots$ so we can compare $\vPat$ with the suffix pointed from any suffix array position along the binary search in $\Oh(\vsPat)$ time, and thus find the suffix array range $\SA[\vSPosPat..\vEPosPat]$ in time $\ftSearch = \Oh(\vsPat\log\vsTColl)$.

For locating, we must be able to compute any $\SA[i]$. The $\CSA$ stores the $\SA$ values that point to text positions that are a multiple of $s$, for a space-time tradeoff parameter $s$. A bitvector $B[1..n]$ with $\rank$ support is used to mark with $B[i]=1$ the sampled positions $\SA[i]$, so a sampled entry $\SA[i]$ is stored at position $\rank(B,i)$ of a sampled array. If $\SA[i]$ is not sampled, the $\CSA$ tries $\SA[\Psi(i)]$, $\SA[\Psi^2(i)]$, and so on, until it finds a sampled $\SA[\Psi^k(i)]$ (i.e., $B[\Psi^k(i)]=1$), for some $k<s$. It then holds that $\SA[i] = \SA[\Psi^k(i)]-k$, so the $\CSA$ supports $\ftLookup = \Oh(s)$. The $\CSA$ then searches in time $\Oh(\vsPat \log \vsTColl + s \cdot \occ)$.

Compression is obtained thanks to the regularities of permutation $\Psi$. For example, because $\Psi$ is increasing in the area of the suffixes starting with the same symbol, it can be represented within the zero-order statistical entropy of $\vTColl$, while supporting constant-time access \cite{Sad03} (more complex $\Psi$-based CSAs obtain higher-order entropy space \cite{GrossiV05:CSA}). To this space, we must add the $\Oh(\vsTColl)$ bits for bitvectors $D$ and $B$, and the $\lfloor \vsTColl/s\rfloor\log \vsTColl$ bits for the samples of $\SA$. The text $\vTColl$ is {\em not} stored.

M\"akinen and Navarro \cite{MakinenN05:SSA} observed another regularity of $\Psi$: it features {\em runs} of consecutive values, that is, $\Psi(i+1)=\Psi(i)+1$. They designed the so-called {\em Run-Length CSA}, or $\RLCSA$, which aimed to use $\Oh(r_\Psi)$ space, where $r_\Psi$ is the number of maximal runs in $\Psi$. It was soon noted that $r_\Psi$ is particularly small on repetitive text collections, which enabled space reductions that are much more significant than those obtained via statistical entropy \cite{MakinenNSV10:SRH}.

Function $\Psi$ was represented in $\Oh(r_\Psi \log \vsTColl)$ bits by encoding the runs $\Psi(i..i+l) = \Psi(i),\Psi(i)+1,\ldots,\Psi(i)+l$ as the pair $\langle \Psi(i),l\rangle$.
The time to access $\Psi$ increases, using modern predecessor data structures \cite{BelazzouguiN15:OLU}, to $\Oh(\log\log_w(\vsTColl/r_\Psi))$, where $w$ is the size in bits of the computer word (we give the details in Section~\ref{subsec:rcsa-counting}). By also representing bitvectors $D$ and $B$ with predecessor data structures, the $\RLCSA$ searches in time $\Oh((m\log \vsTColl + s \cdot \occ) \log\log_w(\sigma+s+\vsTColl/r_\Psi))$.
The total space of the $\RLCSA$ is then $\Oh((r_\Psi + \vsTColl/s)\log n)$ bits. In highly repetitive text collections, the term $\vsTColl/s$ overshadows $r_\Psi$ and ruins the high compression achieved by collapsing the runs in $\Psi$.

\subsubsection{$\BWT$--based CSAs}
The \emph{Burrows--Wheeler Transform}~\cite{BurrowsW94:BSL} of $\vTColl$ is a permutation $\BWT[1..\vsTColl]$ of the symbols of $\vTColl[1..\vsTColl]$ defined as $$\BWT[i]=\vTColl[\SA[i]-1]$$
(and $\vTColl[\vsTColl]=\$$ if $\SA[i]=1$), which boosts the compressibility of $\vTColl$.
The $\FMIdx$~\cite{FerraginaM05:ICT, FerraginaMMN07:CRS} is a CSA that represents $\SA$ and $\vT$ within the high-order statistical entropy of $\vT$, by exploiting the connection between the $\BWT$ and $\SA$.
For counting, the $\FMIdx$ resorts to \emph{backward search}, which successively finds the suffix array ranges $\SA[\vsp_i..\vep_i]$ of $\vP[i..\vm]$, for $i=\vm$ to $1$, starting from $\SA[\vsp_{\vm+1}..\vep_{\vm+1}] = [1..\vn]$ and then
\begin{align*}
    \vsp_i &= C[c] + \rank_c(\BWT,\vsp_{i+1}-1)+1, \\
    \vep_i &= C[c] + \rank_c(\BWT,\vep_{i+1}),
\end{align*}
where $c=\vP[i]$, $C[c]$ is the number of occurrences of symbols smaller than $c$ in $\vT$, and $\rank_c(\BWT,j)$ is the number of times $c$ occurs in $\BWT[1..j]$.
Thus, $[\vsp,\vep]=[\vsp_1,\vep_1]$ if $\vsp_i \le \vep_i$ holds for all $1 \le i \le \vsPat$, otherwise $\vPat$ does not occur in $\vTColl$.

For locating the occurrences $\SA[\vSPosPat],\ldots, \SA[\vEPosPat]$, the $\FMIdx$ samples $\SA$ just like the $\CSA$. The function used to traverse the text towards a sampled position is the so-called {\em $\LF$-step}, which simulates a backward traversal of $\vTColl$: if $\SA[i]=j$, the value $i'$ such that $\SA[i']=j-1$ is $\LF(i)$, where
\[
    \LF(i) = C[c] + \rank_c(\BWT,i),
\] 
where $c=\BWT[i]$.
Note that $\LF(i)$ is the inverse function of $\Psi(i)$. Starting from $\SA[i]$, we compute $\LF$ successively until, for some $0 \le k < s$, we find a sampled entry $\SA[\LF^k(i)]$, which is stored explicitly. It then holds $\SA[i] = \SA[\LF^k(i)]+k$.

By implementing $\BWT$ with a wavelet tree \cite{GGV03}, for example, access and $\rank_c$ on $\BWT$ can be supported in time $\Oh(\log\sigma)$, and the $\FMIdx$ searches in time $\Oh((m+s\cdot occ)\log\sigma)$~\cite{FerraginaMMN07:CRS}. With more sophisticated wavelet tree representations \cite{MN08,KP11}, the space of the $\FMIdx$ is the high-order entropy of $\vTColl$ plus the $\Oh((n/s)\log n)$ bits for the sampling of $\SA$.

The \emph{Run-Length FM-index}, $\RLFMIdx$ \cite{MakinenN05:SSA, MakinenNSV10:SRH} is an adaptation of the $\FMIdx$ aimed at repetitive texts, just like the $\RLCSA$ is to the $\CSA$. Say that the $\BWT[1..\vsTColl]$ is formed by $\vnBWTRuns$ maximal {\em runs of equal symbols}, then it holds that  $\vnBWTRuns$ is small in repetitive collections (in particular, it holds $r_\Psi \le r \le r_\Psi+\sigma$ \cite{MakinenN05:SSA}). For example, it is now known that $\vnBWTRuns = \Oh(z\log^2 \vsTColl)$, where $z$ is the number of phrases of the Lempel--Ziv parse of $\vTColl$~\cite{KK19}.

The $\RLFMIdx$ supports counting within $\Oh(\vnBWTRuns\log n)$ bits, by implementing the backward search over data structures that use space proportional to the number of $\BWT$ runs.
It marks in a bitvector $\Start[1..\vn]$ with $1$s the positions $i$ starting $\BWT$ runs, that is, where $i=1$ or $\BWT[i] \not= \BWT[i-1]$.
The first letter of each run is collected in an array $\Letter[1..\vr]$.
Since $\Start$ has only $r$ $1$s, it can be represented within $\vr\log(\vn/\vr)+\Oh(\vr)$ bits, so that any bit $\Start[i]$ and $\rank(\Start,i)$
are computed in time $\Oh(\log(\vn/\vr))$~\cite{OS07}.
We then simulate $\BWT[j]=\Letter[\rank(\Start,j)]$ in $\Oh(\vr\log n)$ bits.
The backward search formula can be efficiently simulated as well, by adding another bitvector that records the run lengths in lexicographic order. Overall, the search time becomes $\Oh((m+s\cdot occ)\log(\sigma+n/r))$ (by replacing the sparse bitvectors with predecessor data structures and using an alternative to wavelet trees \cite{GMR06}, one can reach $\Oh((m+s\cdot occ)\log\log(\sigma+n/r))$).
The $\RLFMIdx$ still uses $\SA$ samples to locate, however, and when $\vnBWTRuns \ll \vsTColl$ (i.e., on repetitive texts), the $\Oh((\vsTColl / s)\log \vn)$ added bits ruin the $\Oh(\vr\log\vn)$-bit space (unless one accepts high locating times by setting $s \approx r$).

The $\RIdx$~\cite{GagieNP20:FFS} closed the long-standing problem of efficiently locating the pattern occurrences using $\Oh(\vnBWTRuns\log n)$-bit space.
The experiments showed that the $\RIdx$ outperforms all the other implemented indexes by orders of magnitude in space or in search time on highly repetitive datasets.
However, other experiments on more typical repetitiveness scenarios~\cite{NS19,CNP21,DN21,BCGHMNR21} showed that the space of the $\RIdx$ degrades very quickly as repetitiveness decreases.
For example, a grammar-based index (which can be of size $g = \Oh(z\log(\vsTColl/z))$) is usually slower but significantly smaller~\cite{CNP21}, and an even slower Lempel--Ziv based index of size $\Oh(z)$~\cite{KreftN13:CIR} is even smaller.
Some later proposals~\cite{NT20} further speed up the $\RIdx$ by increasing the constant accompanying the $\Oh(\vr\log\vn)$-bit space.
The unmatched time performance of the $\RIdx$ comes then with a very high price in space on all but the most highly repetitive text collections, which makes it of little use in many relevant application scenarios.
This is the problem we address in this paper.

\iftoggle{rindex}{
\section{The \texorpdfstring{$\RIdx$}{r-index} Sampling Mechanism} \label{sec:rindex}

Gagie et al.~\cite{GagieNP20:FFS} provide an $\Oh(\vnBWTRuns\log\vn)$-bits data structure that not only finds the range $\SA[\vSPosPat.. \vEPosPat]$ of the occurrences of $\vPat$ in $\vTColl$, but also gives the value $\SA[\vEPosPat]$, that is, the text position of the last occurrence in the range.
They then provide a second $\Oh(\vnBWTRuns\log n)$-bits data structure that, given $\SA[j]$, efficiently finds $\SA[j-1]$.
This suffices to efficiently find all the occurrences of $\vPat$, in time $\Oh((m+occ)\log\log(\sigma+n/r))$ in their theoretical version.

In addition to the theoretical design, Gagie et al.\ and Boucher et al.~\cite{GagieNP20:FFS,BGKLMM19} provided a carefully engineered $\RIdx$ implementation.
The counting data structures (which find the range $\SA[\vSPosPat.. \vEPosPat]$) require, for any small constant $\epsilon>0$, $\vr \cdot ((1+\epsilon)\log(\vn/\vr)+\log\sigma+\Oh(1))$ bits (largely dominated by the described arrays $\Start$ and $\Letter$), whereas the locating data structures (which obtain $\SA[\vEPosPat]$, and $\SA[j-1]$ given $\SA[j]$), require $\vr \cdot (2\log \vn + \Oh(1))$ further bits.
The locating structures are then significantly heavier in practice, especially when $\vn/\vr$ is not that large.
Together, the structures use $\vr \cdot ((1+\epsilon)\log(\vn/\vr)+2\log \vn + \log\sigma+\Oh(1))$ bits of space and perform backward search steps and $\LF$-steps in time $\Oh(\frac{1}{\epsilon}\log(\sigma+\vn/\vr))$, so they search for $\vP$ in time $\Oh(\frac{1}{\epsilon}(\vm+\occ)\log(\sigma+\vn/\vr))$.

For conciseness, we do not describe the counting data structures of the $\RIdx$, which are the same of the $\RLFMIdx$ and which we do not modify in our index.
The $\RIdx$ locating structures, which we do modify, are formed by the following components:

\begin{description}
\item[{$\First[1..\vsTColl]$}:] a bitvector marking with $1$s the {\em text} positions of the letters that are the first in a $\BWT$ run.
That is, if $j=1$ or $\BWT[j] \not= \BWT[j-1]$, then $\First[\SA[j]-1]=1$.
Since $\First$ has only $\vr$ $1$s, it is represented in compressed form using $\vr\log(\vn/\vr)+\Oh(\vr)$ bits, while supporting $\rank_1$ in time $\Oh(\log(\vn/\vr))$ and, in $\Oh(1)$ time, the operation $\select_1(\First,j)$ (the position of the $j$th $1$ in $\First$)~\cite{OS07}.
This allows one find the rightmost $1$ up to position $i$, $\predecessor(\First,i)=\select_1(\First,\rank_1(\First,i))$.

\item[{$\FirstToRun[1..\vnBWTRuns]$}:] a vector of integers
(using $\vr\lceil \log \vnBWTRuns \rceil$ bits) mapping each letter marked in $\First$ to the $\BWT$ run where it lies.
That is, if the $p$th $\BWT$ run starts at $\BWT[j]$, and $\First[i]=1$ for $i=\SA[j]-1$, then $\FirstToRun[\rank_1(\First,i)]=p$.

\item[{$\Samples[1..\vnBWTRuns]$}:] a vector of $\lceil \log \vn\rceil$-bit integers storing samples
of $\SA$, so that $\Samples[p]$ is the text position $\SA[j]-1$ corresponding to the last letter $\BWT[j]$ in the $p$th $\BWT$ run.
\end{description}

These structures are used in the following way in the $\RIdx$ implementation~\cite{GagieNP20:FFS}:

\begin{description}
\item[Problem 1:] When computing the ranges $\SA[\vSPosPat.. \vEPosPat]$ along the backward search, we must also produce the value $\SA[\vEPosPat]$.
They actually compute all the values $\SA[\vep_i]$.
This is stored for $\SA[\vep_{\vm+1}]=\SA[n]$ and then, if $\BWT[\vep_{i+1}]=\vP[i]$, we know that $\vep_i=\LF(\vep_{i+1})$ and thus $\SA[\vep_i] = \SA[\vep_{i+1}]-1$.
Otherwise, $\vep_i=\LF(j)$ and $\SA[\vep_i]=\SA[j]-1$, where $j \in [\vsp_{i+1}..\vep_{i+1}]$ is the largest position with $\BWT[j]=\vP[i]$.
The position $j$ is efficiently found with their counting data structures, and the remaining problem is how to compute $\SA[j]$.
Since $j$ must be an end of run, however, this is simply computed as $\Samples[p]+1$, where
$p=\rank_1(\Start,j)$ is the run where $j$ belongs.

\item [Problem 2:] When locating we must find $\SA[j-1]$ from $i=\SA[j]-1$.
There are two cases:
\begin{itemize}
    \item $j-1$ ends a $\BWT$ run, that is, $\Start[j]=1$, and then $\SA[j-1]=\Samples[p-1]+1$, where $p$ is as in Problem 1;

    \item $j-1$ is in the same $\BWT$ run of $j$, in which case they compute $\SA[j-1]=\phi(i)$, where%
\footnote{The special case where $\rank_1(\First,i)=0$ is handled separately.}
\vspace*{-5mm}
\begin{equation} \label{eq:rindex}
\phi(i) = \Samples[\FirstToRun[\rank_1(\First,i)]-1]+1+(i-\predecessor(\First,i)).
\end{equation}
\end{itemize}
\end{description}

\begin{figure}[t]
\begin{center}
\includegraphics[width=\textwidth]{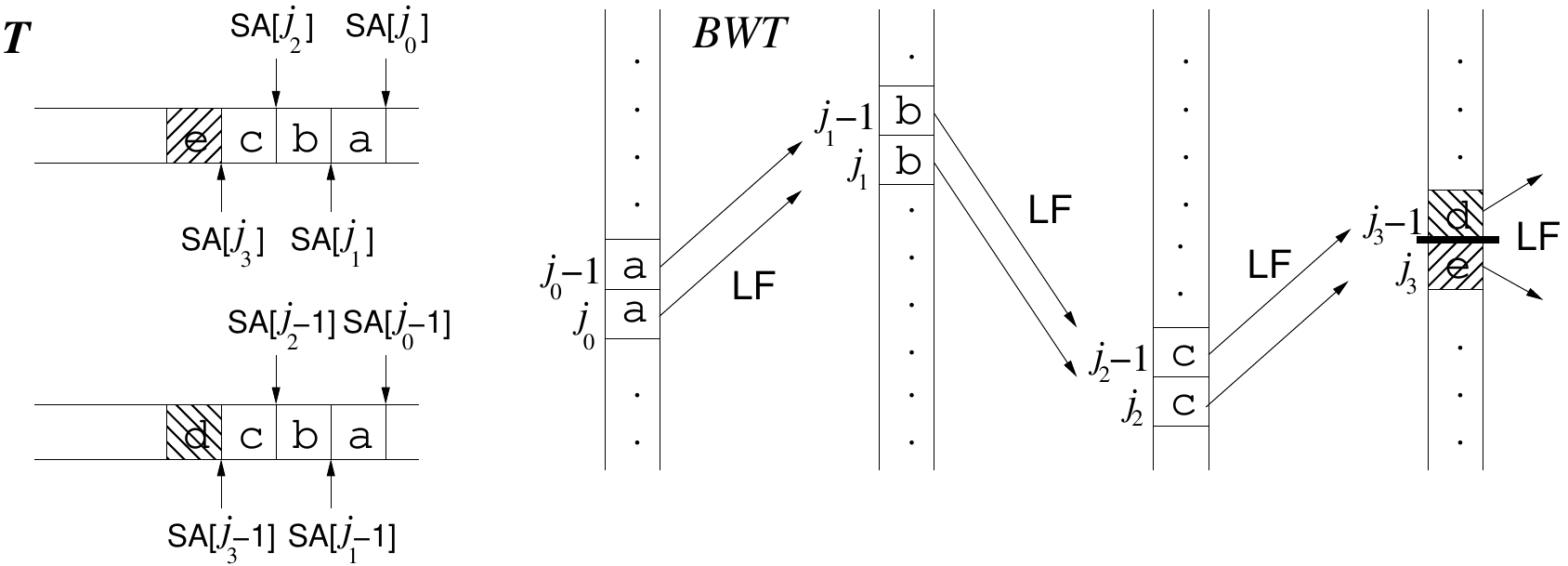}
\end{center}
\caption{Schematic example of the sampling mechanism of the $\RIdx$.
There is a run border between $j_3-1$ and $j_3$.}
\label{fig:figura}
\end{figure}

This formula works because, when $j$ and $j-1$ are in the same $\BWT$ run, it holds that $\LF(j-1)=\LF(j)-1$~\cite{FerraginaM05:ICT}.
Figure~\ref{fig:figura} explains why this property makes the formula work.
Consider two $\BWT$ positions, $j=j_0$ and $j' = j-1 = j_0-1$, that belong to the same run.
The $\LF$ formula will map them to consecutive positions, $j_1$ and $j_1'=j_1-1$.
If $j_1$ and $j_1-1$ still belong to the same run, $\LF$ will map them to consecutive positions again, $j_2$ and $j_2'=j_2-1$, and once again, $j_3$ and $j_3'=j_3-1$.
Say that $j_3$ and $j_3-1$ do not belong to the same run.
This means that $j_3-1$ ends a run (and thus it is stored in $\Samples$) and $j_3$ starts a run (and thus $\SA[j_3]-1$ is marked in $\First$).
To the left of the $\BWT$ positions we show the areas of $\vT$ virtually traversed as we perform consecutive $\LF$-steps.
Therefore, if we know $i=\SA[j]-1=\SA[j_0]-1$, the nearest $1$ in $\First$ to the left is at $\predecessor(\First,i)=\SA[j_3]-1$ (where there is an \texttt{e} in $\vT$) and $p=\FirstToRun[\rank(i)]$ is the number of the $\BWT$ run that starts at $j_3$.
If we subtract $1$, we have the $\BWT$ run ending at $j_3-1$, and then $\Samples[p-1]$ is the position preceding $\SA[j_3-1]$ (where there is a \texttt{d} in $\vT$).
We add $1+(i-\predecessor(\First,i))=4$ to obtain $\SA[j_0-1]=\SA[j-1]$.

These components make up, effectively, a sampling mechanism of $\Oh(\vr\log n)$ bits (i.e., sampling the end of runs), instead of the traditional one of  $\Oh((n/s)\log n)$ bits (i.e., sampling every $s$th text position).

%
%
}

\section{\texorpdfstring{\RCSA}{r-CSA}: A \texorpdfstring{$\Psi$}{Psi}-based Index for Repetitive Texts} \label{sec:rcsa}

In this section we introduce the $\RCSA$, an equivalent to the $\RIdx$ based on the $\RLCSA$. We first describe the counting algorithm and data structures, which is just a modern version of those given in the original $\RLCSA$ \cite{MakinenNSV10:SRH}. We then show how to locate within $\Oh(r)$ space, by translating the techniques of the $\RIdx$ \cite{GagieNP20:FFS} to this scenario.

%

\subsection{Counting in $\Oh(r)$ space}\label{subsec:rcsa-counting}

Let $\Psi_{c} = \Psi[i..j]$ be the range in $\Psi$ corresponding to each symbol $c \in [1..\sigma]$, such that all the suffixes of $\vTColl$ starting with $c$ are in the range $\SA[i..j]$. As said,
$\Psi_{c}$ is strictly increasing over $[1..\vsTColl]$; let us say that it contains $r_{c}$ maximal runs of consecutive values.
By definitions of $\Psi$ and $\BWT$, if $\Psi_{c}(i) = k$ then $\BWT[k] = \vTColl[\SA[i]] = c$.
Consequently, a one-to-one relation exists between the $p$th $\Psi_{c}$ run and the $p$th $\BWT$ run of symbol $c$, for $1 \le p \le \vnBWTRuns_{c}$.
It then follows that $\vnBWTRuns = \sum_{c=1}^{\sigma} \vnBWTRuns_{c}$.\footnote{
  Some sequences of consecutive values in $\Psi$ can extend beyond the limit of the corresponding range $\Psi_{c}$. This is why it holds $r_\Psi \le r$, which becomes an equality if we split those runs of $\Psi$ by allowing only runs inside each $\Psi_c$.} We call $\Psi$-runs those $r$ maximal runs in $\Psi$ that are inside some range $\Psi_c$.

A backward search process was also devised for the $\CSA$ \cite{Sad02} and used in the $\RLCSA$ \cite{MakinenNSV10:SRH}. Once the range $\SA[sp_{i+1}..ep_{i+1}]$ for $P[i+1..m]$ is known, we obtain $\SA[sp_i..ep_i]$ by binary searching within $\Psi_c$, for $c=P[i]$, the maximal range of positions $j$ such that $\SA[j] \in [sp_{i+1}..ep_{i+1}]$. Indeed, those are the suffixes that start with $c=P[i]$ and follow with $P[i+1..m]$. This technique yields the same $\Oh(m\log n)$ counting time, but has better locality of reference.

As in the the $\RLCSA$, we represent each $\Psi$-run $\Psi(i..i+l)$ as a pair $\langle \Psi(i), l \rangle$. Unlike the original $\RLCSA$, we construct a predecessor data structure $\mathcal{P}_{\Psi}$ on the $\vnBWTRuns$ $\Psi$-run heads within the universe $\sigma\vsTColl$ by concatenating all the $\sigma$ ranges $\Psi_{c}$.
Specifically, the predecessor $\mathcal{P}_{\Psi}$ stores $x + (c-1)\vsTColl$ if $x=\Psi(i)$ is a $\Psi$-run head in $\Psi_{c}$, allowing us to compute the $\Psi$-run that contains or precedes a given value $y \in [1..\vsTColl]$.
Thus, the predecessor operation $\predecessor_c$ on $\Psi_c$ is defined in terms of a classic predecessor function $\predecessor$ on the $\Psi$-run heads, as
\[\predecessor_c(y) = \predecessor(\mathcal{P}_{\Psi}, y + (c-1)\vsTColl) = \langle x, k \rangle\]
where $x$ is the actual predecessor value in $\mathcal{P}_{\Psi}$, and $k$ is its $\Psi$-run rank (i.e., the number of $\Psi$-runs up to the one that starts with value $x$).
Using recent predecessor structures \cite[Thm. A.1]{BelazzouguiN15:OLU} to represent $\mathcal{P}_{\Psi}$, we use $\Oh(\vnBWTRuns \log(\vsTColl/\vnBWTRuns))$ bits of space and answer queries in $\Oh(\log \log_{w}(\vsAlph \vsTColl / \vnBWTRuns))$ time.

In addition, we associate each $\Psi$-run head $x=\Psi(i)$ with its global position $i$ in $\Psi$ in an array $I_{\Psi}[1..\vnBWTRuns]$, where $I_{\Psi}[j] = i$ iff $\Psi(i)$ is the first item of the $j$th $\Psi$-run.
$I_{\Psi}$ is used to support the backward search, computing the length of $\Psi$-runs and of each new range in $\SA$.
Structure $I_{\Psi}[1..\vnBWTRuns]$ replaces the array $C[1..\sigma]$ of the original $\RLCSA$~\cite{MakinenNSV10:SRH}.
Figure~\ref{fig:rcsa-counting} illustrates definitions and relations.


\begin{figure}[t]
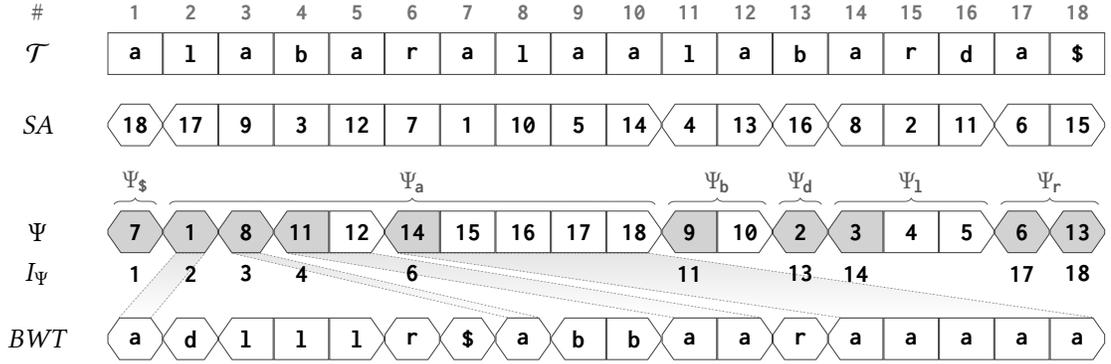

  \begin{center}
    \begin{tikzpicture}[font=\ttfamily \bfseries \small, transform shape]

\input{tikz_styles}
\tikzstyle{CellHeader} = [CellHeaderBase]
\tikzstyle{Cell} = [CellBase]

\begin{scope}[shift=({0,+3.5ex})]
  \tikzstyle{CellHeader} = [CellHeaderBase, black!65]
  \tikzstyle{Cell} = [CellBase, draw=none, font=\ttfamily \bfseries \footnotesize, black!55]
  \input{base-pos}
\end{scope}

\begin{scope}[shift=({0,0})]
  \input{base-text}
\end{scope}

\begin{scope}[shift=({0,-6ex})]
  \input{base-sa}
\end{scope}

\begin{scope}[shift=({0,-15ex})]
  \input{rcsa-counting-psi}
\end{scope}

\begin{scope}[shift=({0,-18.5ex})]
  \tikzstyle{Cell} = [CellBase, draw=none]
  \input{base-i-psi}
\end{scope}

\begin{scope}[shift=({0,-24ex})]
  \input{base-bwt}
\end{scope}

\begin{scope}[on background layer]
  \tikzstyle{Stripe} = [-, densely dotted, gray, shading=axis, top color=gray!20, bottom color=white]
  \filldraw[Stripe]
  (BWT-0.before north west) to (Psi-1.after south west) to (Psi-1.before south east) to (BWT-0.after north east);
  \filldraw[Stripe]
  (BWT-7.before north west) to (Psi-2.after south west) to (Psi-2.before south east) to (BWT-7.after north east);
  \filldraw[Stripe]
  (BWT-10.before north west) to (Psi-3.after south west) to (Psi-4.before south east) to (BWT-11.after north east);
  \filldraw[Stripe]
  (BWT-13.before north west) to (Psi-5.after south west) to (Psi-9.before south east) to (BWT-17.after north east);
\end{scope}

\end{tikzpicture}
  \end{center}
  \caption{
    Data structures for the counting mechanism of the $\RCSA$ on an example text.
    The blocks in $\SA$ cover the suffixes starting with the same symbol, which align with the areas of each $\Psi_c$. The blocks in $\Psi$ and $\BWT$ represent runs.
    The gray cells in $\Psi$ are the $\vnBWTRuns$ run heads.
    The stripes show the relation between each run of $\Psi_\mathsf{a}$ and its corresponding $\BWT$ run.
  }
  \label{fig:rcsa-counting}
\end{figure}

Algorithm~\ref{alg:rcsa_counting} computes the suffix array range $A[sp..ep]$ of the occurrences of a pattern $\vPat$ in the text $\vTColl$, using the predecessor structure $\mathcal{P}_{\Psi}$ and the array $I_{\Psi}$. Note that it must consider the cases where the answer is within a $\Psi$-run or not. This yields the following result. 

\begin{algorithm}[t]
  \IO{Query pattern $\vPat[1..\vsPat]$.}
  {Range $\langle \vsp,\vep \rangle$ on $\SA$ for $\vPat$.}

  \begin{multicols}{2}
    \SetKwFunction{fnCount}{count}
\SetKwFunction{fnNextSP}{findNextStartPos}
\SetKwFunction{fnNextEP}{findNextEndPos}

\Function{\fnCount{$\vPat[1..\vsPat]$}}{
  $\langle \vsp, \vep \rangle \leftarrow \langle 1, \vsTColl \rangle$\;

  \For{$i \leftarrow \vsPat$ \KwDownTo $1$}{
    $c \leftarrow \vPat[i]$\;

    $\vsp \leftarrow \fnNextSP(\vsp, c)$\;
    $\vep \leftarrow \fnNextEP(\vep, c)$\;
    \If{$\vsp > \vep$} {\Return{``$\vPat$ is not in $\vTColl$''}}
  }

  \Return{$\langle \vsp, \vep \rangle$}\;
}
\BlankLine
\BlankLine

\Function{\fnNextSP{$\vsp$, $c$}}{\label{fn:NextSP}
  $\vsp' \leftarrow \vsp + (c-1)\cdot\vsTColl$\;
  $\langle x, k \rangle \leftarrow \predecessor(\mathcal{P}_{\Psi}, \vsp')$\;
  \If{$\vsp' < x + (I_{\Psi}[k+1]-I_{\Psi}[k])$}{\Return{$I_{\Psi}[k] + (\vsp' - x)$}}
  \Return{$I_{\Psi}[k+1]$}\;
}
\BlankLine
\BlankLine

\Function{\fnNextEP{$\vep$, $c$}}{
  $\vep' \leftarrow \vep + (c-1)\cdot\vsTColl$\;
  $\langle x, k \rangle \leftarrow \predecessor(\mathcal{P}_{\Psi}, \vep')$\;
  \If{$\vep' < x + (I_{\Psi}[k+1]-I_{\Psi}[k])$}{\Return{$I_{\Psi}[k] + (\vep' - x)$}}
  \Return{$I_{\Psi}[k+1] - 1$}\;
}

  \end{multicols}
  \BlankLine

  \caption{Counting pattern occurrences with $\RCSA$.}
  \label{alg:rcsa_counting}
\end{algorithm}

\begin{theorem}[Counting with $\RCSA$]\label{thm:rcsa_counting}
  The $\RCSA$ of a text $\vTColl[1..\vsTColl]$ over alphabet $[1..\sigma]$, with $\vnBWTRuns$ $\Psi$-runs, can be represented using $\Oh(\vnBWTRuns \log \vsTColl)$ bits of space and count the occurrences of any pattern $\vPat[1..\vsPat]$ in $\Oh(\vsPat \log \log_{w}(\vsAlph \vsTColl / \vnBWTRuns))$ time, where $w$ is the computer word size.
\end{theorem}

\subsection{Locating in $\Oh(r)$ space}\label{subsec:rcsa-locating}

Following the $\RIdx$ \cite{GagieNP20:FFS}, we reduce the problem of locating the occurrences of pattern $\vPat$ with the $\RCSA$ to two subproblems: (1) maintaining the text position of $\SA[sp_i]$ along the backward search, (2) finding $\SA[j+1]$ given $\SA[j]$. After the backward search, then, we know $\SA[sp]$ by (1) and then find $\SA[sp+1], \SA[sp+2], \ldots, \SA[ep]$ with (2).

\subsubsection{Counting with toehold}\label{subsubsec:rcsa-problem-1}

In the same vein as the $\RIdx$ \cite[Lem.~3.2]{GagieNP20:FFS}, we show how to enhance the backward search so that we always know $\SA[sp_i]$ (called the ``toehold''). We give a proof that this can be done that is better suited for practical $\Psi$-based indexes.

\begin{lemma}\label{lem:r-csa-toehold}
  The $\RLCSA$ backward search process on $\Psi$ can be enhanced to retrieve, along with the range $\langle \vSPosPat, \vEPosPat \rangle$ on $\SA$ for the pattern $\vPat[1..\vsPat]$, the toehold value $\SA[\vSPosPat]$ in $\Oh(1)$ additional time per backward step and with $\Oh(\vnBWTRuns \log \vsTColl)$ additional bits of space.
\end{lemma}

\begin{proof}
  We store in a new array $F_{\SA}[1..\vnBWTRuns]$ the text positions of the $\Psi$-run heads, that is, $F_{\SA}[j] = \SA[i]$ iff the $j$th $\Psi$-run begins at position $i$.
  The backward search initiates with the entire interval $\langle \vSPosPat_{\vsPat+1}, \vEPosPat_{\vsPat+1} \rangle = \langle 1, \vsTColl \rangle$ of $\SA$.
  The initial toehold is then $\SA[\vSPosPat_{\vsPat+1}] = \SA[1] = F_{\SA}[1] = \vsTColl$.

  Let $\langle \vSPosPat_{i+1}, \vEPosPat_{i+1} \rangle$ be the range on $\SA$ for the occurrences of $\vPat[i+1..\vsPat]$, with $\SA[\vSPosPat_{i+1}]$ being a known value.
  As described in Algorithm~\ref{alg:rcsa_counting}, the function $\fnNextSP$ relies on the operation $\predecessor(\mathcal{P}_{\Psi}, \vSPosPat_{i+1}, \vPat[i]) = \langle x, k \rangle$ to compute the starting position $\vSPosPat_{i}$ of the interval for $\vPat[i..\vsPat]$.

  Extending $\fnNextSP$ to additionally obtain the value $\SA[\vSPosPat_{i}]$ when the pattern occurs in $\vTColl$ results in two possible cases.
  If the $k$th $\Psi$-run contains the value $\vSPosPat_{i+1}$, then $\vSPosPat_{i+1} = \Psi(\vSPosPat_{i})$ because the first value of the range $\SA[sp_{i+1}..ep_{i+1}]$ is preceded by $P[i]$. Thus the next toehold is straightforwardly calculated as $\SA[\vSPosPat_{i}] = \SA[\vSPosPat_{i+1}]-1$.
  If, instead, $\vSPosPat_{i+1}$ does not belong to the $k$th $\Psi$-run, then $\vSPosPat_{i}$ is the head of the $(k+1)$th $\Psi$-run.
  Using the $F_{\SA}$ array, the next toehold is computed as $\SA[\vSPosPat_{i}] = F_{\SA}[k+1]$, also in constant time.
\end{proof}


\begin{algorithm}[t]
  \IO{Query pattern $\vPat[1..\vsPat]$.}
  {Range $\langle \vsp,\vep \rangle$ on $\SA$ for $\vPat$;
  Value $\SA[\vsp]$.}

  \begin{multicols}{2}
    \SetKwFunction{fnCount}{count}
\SetKwFunction{fnNextSP}{findNextStartPos}
\SetKwFunction{fnNextEP}{findNextEndPos}

\Function{\fnCount{$\vPat[1..\vsPat]$}}{
  $\langle \vsp, \vep, v \rangle \leftarrow \langle 1, \vsTColl, n \rangle$\;
  \For{$i \leftarrow \vsPat$ \KwDownTo $1$}{
    $c \leftarrow \vPat[i]$\;

    $\langle v,\vsp\rangle \leftarrow \findStartToehold(v, \vsp, c)$\;
    $\vep \leftarrow \fnNextEP(\vep, c)$\;
    \If{$\vsp > \vep$} {\Return{``$\vPat$ is not in $\vTColl$''}}
  }

  \Return{$\langle \vsp, \vep, v \rangle$}\;
}
\BlankLine
\BlankLine

\Function{\findStartToehold{($v$, $\vsp$, $c$)}}{\label{fn:NextSP2}
  $\vsp' \leftarrow \vsp + (c-1)\cdot\vsTColl$\;
  $\langle x, k \rangle \leftarrow \predecessor(\mathcal{P}_{\Psi}, \vsp')$\;
  \If{$\vsp' < x + (I_{\Psi}[k+1]-I_{\Psi}[k])$}{\Return{$\langle v-1,I_{\Psi}[k] + (\vsp' - x)\rangle$}}
  \Return{$\langle F_{\SA}[k+1],I_{\Psi}[k+1]\rangle$}\;
}

  \end{multicols}
  \BlankLine

  \caption{Counting pattern occurrences and finding value $\SA[\vsp]$ with $\RCSA$.}
  \label{alg:rcsa_count_toehold}
\end{algorithm}
Algorithm~\ref{alg:rcsa_count_toehold} gives the corresponding pseudocode.

\subsubsection{Locating from toehold}\label{subsubsec:rcsa-problem-2}

While it is possible to employ a sampling scheme nearly identical to that utilized by the $\RIdx$, we opt for an alternative one that virtually moves forwards in the text, rather than backwards. This choice is influenced by the nature of the $\Psi$ function, which enables more efficient forward than backward traversal. 

\begin{lemma} \label{lem:sa_fw}
Given a text position $\SA[i]$,
let $\Psi[l]$ be the tail of the $\Psi$-run with the smallest text position $\SA[l] \ge \SA[i]$. Then,
\begin{equation}
  \label{eq:sa_fw}
  \SA[i+1] = \SA[l+1] + (\SA[l] - \SA[i]).
\end{equation}
\end{lemma}

\begin{proof}
  There are two possible cases. The first case, where
%
  {$\Psi[i]$ is the last symbol of a $\Psi$-run}, is trivial because $i=l$.
For the second case, where
  {$\SA[i]$ is not the last symbol of a $\Psi$-run},
  let $\Delta = \SA[l] - \SA[i]$, that is, $l = \Psi^\Delta(i)$. By the definition of $l$, it holds for all $0 \leq p < \Delta$ that
  $\Psi^{p}(i)$ is {\em not} the last element of a $\Psi$-run. Consequently, for all $0 \le p \le \Delta$, it holds that
  \[
    \Psi^{p}(i) + 1 = \Psi^{p}(i+1).
  \]

  This
  shows that each pair $\Psi^{p}(i)$ and $\Psi^{p}(i+1)$ are adjacent positions within a $\Psi$-run until the position $l=\Psi^{\Delta}(i)$ is reached.
  That is, text positions $\SA[\Psi^{p}(i)]$ and $\SA[\Psi^{p}(i+1)]$ traverse forward together in the suffix array for $\Delta$ steps, so
  \[
    \SA[i+1] = \SA[\Psi^{\Delta}(i+1)] - \Delta = \SA[\Psi^\Delta(i)+1] - \Delta =  \SA[l+1] + (\SA[l] - \SA[i]),
  \]
where the first equality holds just by definition of $\Psi$.
  Figure~\ref{fig:rcsa-phi-1} illustrates the proof.
\end{proof}


\begin{figure}[t]
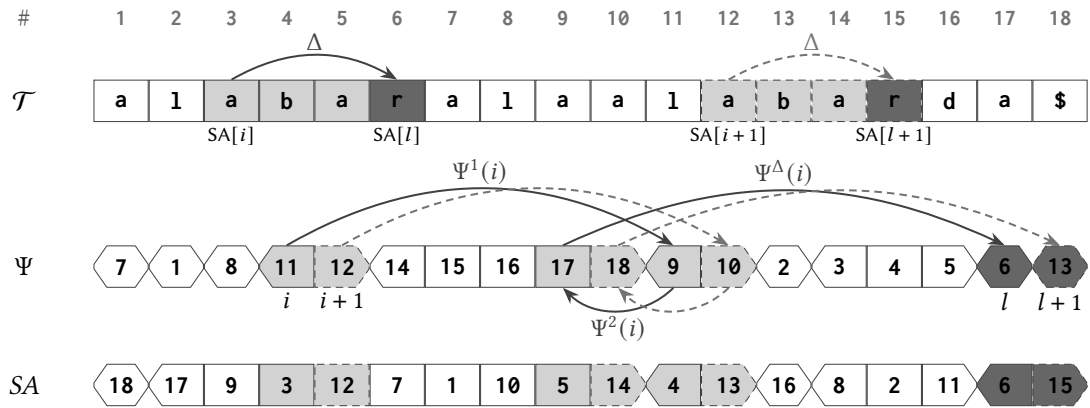

  \begin{center}
    \begin{tikzpicture}[font=\ttfamily \bfseries \small, transform shape]

\input{tikz_styles}
\tikzstyle{CellHeader} = [CellHeaderBase]
\tikzstyle{Cell} = [CellBase]

\begin{scope}[shift=({0,+7ex})]
  \tikzstyle{CellHeader} = [CellHeaderBase, black!65]
  \tikzstyle{Cell} = [CellBase, draw=none, font=\ttfamily \bfseries \footnotesize, black!55]
  \input{base-pos}
\end{scope}

\begin{scope}[shift=({0,0})]
  \input{rcsa-phi-1-text}
\end{scope}

\begin{scope}[shift=({0,-14ex})]
  \input{rcsa-phi-1-psi}
\end{scope}

\begin{scope}[shift=({0,-24ex})]
  \input{rcsa-phi-1-sa}
\end{scope}

\end{tikzpicture}
  \end{center}
  \caption{
    Example of the sampling mechanism of the $\RCSA$.
    The arrows in $\vTColl$ show the $\Delta$ distance from the given $\SA[i]$ to its $\Psi$-run last element successor $\Psi[l]$.
    In $\Psi$, each run is represented by a block;
    solid arrows are $\Psi$ steps for $i$;
    and dashed arrows are $\Psi$ steps for $i+1$.
  }
  \label{fig:rcsa-phi-1}
\end{figure}


K{\"a}rkk{\"a}inen et al.~\cite{KarkkainenMP09} defined the function $\fnPhi$ that returns $\SA[i-1]$ for the given text position $\SA[i]$.
Gagie et al.~\cite{GagieNP20:FFS} later added the function $\fnIPhi$ returning $\SA[i+1]$.
This last function is formally defined below.

\begin{definition}[Gagie et al.~\cite{GagieNP20:FFS}]
  Function $\fnIPhi$ is a permutation of $[1..\vsTColl]$ such that,
  for any text position $j$ and its related position $i$ in the suffix array (i.e., $j = \SA[i]$), is defined as
  \begin{equation}
    \label{eq:phi_inv_def}
    \fnIPhi(j) =
    \begin{cases}
      \SA[\ISA[j] + 1] = \SA[i + 1], & \text{if}\ i < \vsTColl\\
      \SA[1] = \vsTColl, & \text{if}\ i = \vsTColl
    \end{cases}
  \end{equation}
\end{definition}

Gagie et al.~\cite{GagieNP20:FFS} show how to store the permutations $\fnPhi$ and $\fnIPhi$ in $\Oh(\vnBWTRuns\log n)$ bits of space using a predecessor data structure.
We achieve a similar result based on Lemma~\ref{lem:sa_fw}, yet using a successor function over the text positions of $\Psi$-run tails. 

\begin{lemma}
  The function $\fnIPhi$ can be evaluated in $\Oh(\log \log_{w}(\vsTColl / \vnBWTRuns))$ time using an $\Oh(\vnBWTRuns \log \vsTColl)$-bits successor data structure.
\end{lemma}

\begin{proof}
  Let $L$ be the set of $\vnBWTRuns$ text positions such as $x = \SA[l] \in L$ iff $\Psi[l]$ is the last element in its $\Psi$-run,
  and $\mathcal{S}_{L}$ be a successor function over the values of $L$.
  Also, each value $\SA[l] \in \mathcal{S}_{L}$ is paired with the text position $y = \SA[l+1]$ of its next $\Psi$-run head.
  Given a value $j = \SA[i]$,
  if $\left< x, y \right> = \mathcal{S}_{L}(j)$,
  then $\SA[i+1] = y + (x - j)$ by Lemma~\ref{lem:sa_fw}.

  A recent predecessor data structure \cite[Thm. A.1]{BelazzouguiN15:OLU} represents  $\mathcal{S}_{L}$ within $\Oh(\vnBWTRuns \log \vsTColl)$ bits and answers successor queries in time $\Oh(\log \log_{w}(\vsTColl / \vnBWTRuns))$.
\end{proof}

\begin{algorithm}[ht]
  \IO{Query pattern $\vPat[1..\vsPat]$.}
  {Occurrences of $\vPat$: $V[1..\vnOcc] = \SA[\vsp..\vep]$.}

  \SetKwFunction{fnLocate}{locate}
\SetKwFunction{fnCount}{count}

\Function{\fnLocate{$\vPat[1..\vsPat]$}}{
  $\langle \vsp, \vep, v \rangle \leftarrow \fnCount(\vPat)$\;
  $V[1] \leftarrow v$\;

  \For{$i \leftarrow 2$ \KwTo $\vep-\vsp+1$}{
    $V[i] \leftarrow \fnIPhi(V[i-1])$\;
  }

  \Return{$V$}\;
}

  \BlankLine

  \caption{Locating pattern occurrences with $\RCSA$.}
  \label{alg:rcsa_locate}
\end{algorithm}

Algorithm~\ref{alg:rcsa_locate} shows how $\fnIPhi$ is used to compute all the occurrences of $P$ given the first one.
We have now arrived at the primary outcome of this section, which is stated in the following form for compatibility with the $\RIdx$, using that $\Oh(\log(\sigma+n/r)) = \Oh(\log(\sigma n/r))$.

\begin{theorem}[Locating with $\RCSA$] \label{thm:rcsa_locating}
  The $\RCSA$ of a text $\vTColl[1..\vsTColl]$ over alphabet $[1..\sigma]$, with $\vnBWTRuns$ $\Psi$-runs, can be represented within $\Oh(\vnBWTRuns \log \vsTColl)$ bits and locate the $\vnOcc$ occurrences of a pattern $\vPat[1..\vsPat]$ in $\Oh(\vsPat \log \log_{w}(\vsAlph +\vsTColl / \vnBWTRuns) + \vnOcc \log \log_{w}(\vsTColl / \vnBWTRuns))$ time, where $w$ is the computer word size.
\end{theorem}

\subsection{Practical design}\label{subsec:rcsa-practical}

While the theoretical result yields $\Oh(r\log n)$ space without full details on the constants, a finer design is needed in order to obtain a space-competitive data structure, even if it does not yield the same time complexities.
%

Following M{\"a}kinen et al.~\cite{MakinenNSV10:SRH} (and Sadakane \cite{Sad03}), we decompose $\Psi$ into $\vsAlph$ partial functions $\Psi_{c}$, one per symbol $c$.
Because each $\Psi_{c}$ is strictly increasing, we differentially encode the first and last values of the $\Psi$-runs, using $\delta$-codes to represent the differences.
To accelerate access to the function $\Psi$, we sample every $B$-th absolute value, creating a reduced sequence $\widehat{\Psi_{c}}$.
Parameter $B$ yields a tradeoff between space and time to access $\Psi$: one spends $\lfloor r/B \rfloor \log n$ bits on the samples and accesses any cell of $\Psi$ in time $\Oh(B)$, by accessing the preceding cell of $\widehat{\Psi_{c}}$ and then decoding up to $B$ $\delta$-codes.
Further, function $\predecessor(\mathcal{P}_\Psi,i)$ can be computed in time $\Oh(\log(r/B)+B) = \Oh(\log r + B)$, by binary searching the samples $\widehat{\Psi_{c}}$ and then decoding up to $B$ values.

This compressed representation of $\Psi$ requires
\begin{equation*}
  \vnBWTRuns \cdot (\log(\vsAlph \vsTColl / \vnBWTRuns) + \log(\vsTColl / \vnBWTRuns) + \Oh(\log\log(\vsAlph \vsTColl / \vnBWTRuns))) + \Oh((\vnBWTRuns  / B) \log \vsTColl) + \Oh(\vsAlph \log \vsTColl)
\end{equation*}
bits of space.
The initial term is the worst-case size of the run-length encoding of $\Psi$, using $\delta$-codes to store the length of each $\Psi$-run and the gap between them (i.e., the distance between the first value of a $\Psi$-run and the last value of the preceding one).
The second term covers the first value $\Psi(i)$ of every $B$th run, and their absolute ranks.
The last term represents the array $C[1..\vsAlph]$ (used instead of array $I_{\Psi}[1..\vnBWTRuns]$ in our practical proposal), and additional samples of $\Psi$ for the first element in each partial $\Psi_c$.

The second aspect to consider is the practical implementation of function $\fnIPhi$. This relies on several components; Figure~\ref{fig:rcsa-scheme} illustrates their relation.


\begin{figure}[t]
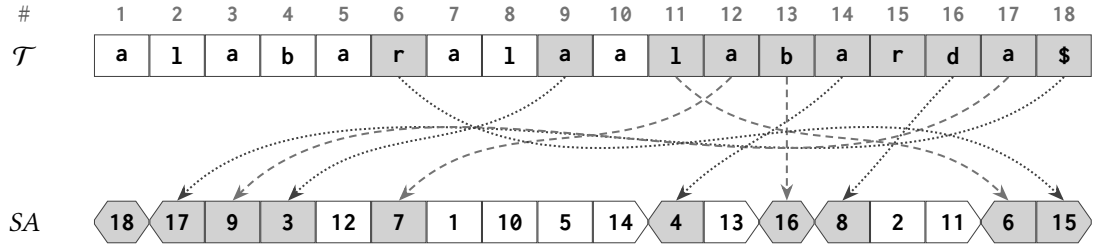

  \begin{center}
    \begin{tikzpicture}[font=\ttfamily \bfseries \small, transform shape]

  \input{tikz_styles}
  \tikzstyle{CellHeader} = [CellHeaderBase]
  \tikzstyle{Cell} = [CellBase]

  \begin{scope}[shift=({0,+3.5ex})]
    \tikzstyle{CellHeader} = [CellHeaderBase, black!65]
    \tikzstyle{Cell} = [CellBase, draw=none, font=\ttfamily \bfseries \footnotesize, black!55]
    \input{base-pos}
  \end{scope}

  \begin{scope}[shift=({0,0})]
    \input{rcsa-scheme-text}
  \end{scope}

  \begin{scope}[shift=({0,-14ex})]
    \input{rcsa-scheme-sa}
  \end{scope}

  \begin{scope}[max distance=9em, out=225, in=45]
    \draw[ArrowDot] (T-17.south) to [out=225, in=45] (SA-1.north);
    \draw[ArrowDash] (T-16.south) to [out=225, in=45] (SA-2.north);
    \draw[ArrowDot] (T-8.south) to [out=225, in=45] (SA-3.north);
    \draw[ArrowDash] (T-11.south) to [out=225, in=45] (SA-5.north);
    \draw[ArrowDot] (T-13.south) to [out=225, in=45] (SA-10.north);
    \draw[ArrowDash] (T-12.south) to [out=270, in=90] (SA-12.north);
    \draw[ArrowDot] (T-15.south) to [out=225, in=45] (SA-13.north);
    \draw[ArrowDash] (T-10.south) to [out=315, in=135] (SA-16.north);
    \draw[ArrowDot] (T-5.south) to [out=315, in=135] (SA-17.north);

  \end{scope}

\end{tikzpicture}
  \end{center}
  \caption{
    Practical data structures used for locating mechanism of the $\RCSA$.
    The gray cells in $\vTColl$ are the $\vnBWTRuns$ elements marked in the bitvector $\MarksL$.
    The gray cells in $\SA$ are the $\vnBWTRuns$ samples stored in the array $\SamplesF$.
    The arrows show the mapping $\MapMarksL$ (different arrow styles are used to improve visualization).
  }
  \label{fig:rcsa-scheme}
\end{figure}

\begin{description}
  \item[{$\MarksL[1..\vsTColl]$}:] a bitvector marking with $\MarksL[j] = 1$ the text positions $j = \SA[i]$ where $\Psi(i)$ is the last symbol of a $\Psi$-run.
  Since $\MarksL$ has only $\vr$ $1$s, it is represented in compressed form using $\vr \log(\vn / \vr) + \Oh(\vr)$ bits, while supporting $\rank(\MarksL, i)$ in time $\Oh(\log(\vn / \vr))$ and, in $\Oh(1)$ time, the operation $\select(\MarksL, p)$ (the position of the $p$th $1$ in $\MarksL$)~\cite{OS07}. 
  This allows one to find the leftmost $1$ from position $i$ as
  \begin{equation*}
    \successor(\MarksL, i) = \select_1(\MarksL, \rank_1(\MarksL, i - 1) + 1).
  \end{equation*}

  \item[{$\MapMarksL[1..\vnBWTRuns]$}:] an array of integers (using $\vr\lceil \log \vnBWTRuns \rceil$ bits) mapping each text position marked in $\MarksL$ to the related sample in $\SamplesF$.
  Note that, if $\MarksL[j] = 1$ with $j = \SA[l]$, then there exists $p$ such that $\SamplesF[p] = \SA[l + 1]$, because $\Psi(l)$ is the last symbol in a $\Psi$-run.\footnote{In the particular case where $l = \vsTColl$ (that is, $j = \SA[l]$ is the last symbol of the final $\Psi$-run), the associated sample in $\SamplesF$ is $p = \mapping(\MarksL, j) = \MapMarksL[\vnBWTRuns] = 1$, which corresponds to the $\Psi$-run of the special symbol \$.}
  We find it with
  \begin{equation*}
    p = \mapping(\MarksL, j) = \MapMarksL[\rank_1(\MarksL, j - 1) + 1].
  \end{equation*}
\end{description}

If we apply this formula on a text position $j$ that does not correspond to the end of a run, it will return the run number $p$ of the next text position that corresponds to the end of a run. Using also the $\SamplesF$ array in addition to the components mentioned above, we can then compute the function $\fnIPhi$ as follows:
\begin{equation} \label{eq:rcsa_iphi}
  \fnIPhi(j) = \SamplesF[\mapping(\MarksL, j)] - (\successor(\MarksL, j) - j).
\end{equation}


A straightforward examination of the preceding data structures reveals that they collectively yield the following result.

\begin{theorem}[Practical $\RCSA$] \label{thm:rcsa_practical}
  The practical $\RCSA$ of a text $\vTColl[1..\vsTColl]$ over alphabet $[1..\sigma]$, with $\vnBWTRuns$ $\Psi$-runs, is represented using
  \begin{equation*}
    \vnBWTRuns \cdot \big(2\log \vsTColl + 2 \log( \vsTColl / \vnBWTRuns) + \log \vsAlph + \Oh(\log\log(\vsAlph \vsTColl / \vnBWTRuns))\big) + \Oh((\vnBWTRuns  / B) \log \vsTColl) + \Oh(\vsAlph \log \vsTColl) \\
  \end{equation*}
  bits and can locate the $\vnOcc$ occurrences of a pattern $\vPat[1..\vsPat]$ in $\Oh(\vsPat (\log\vnBWTRuns + B) + \vnOcc \log(\vsTColl / \vnBWTRuns))$ time, where $B$ is the block size for the representation of $\Psi$.
\end{theorem}

\section{Subsampled \texorpdfstring{$\BWT/\Psi$}{BWT/Ψ}-based indexes} \label{sec:srcsa}

While the $\RIdx$ and $\RCSA$ sampling mechanisms perform very well on highly repetitive texts, they can be less efficient in areas where the run heads or tails split the text into many short blocks. The text is sampled too frequently in such areas, creating unnecessary redundancy in the indexes.
Figure~\ref{fig:samples-densities} illustrates an analysis of commonly used datasets confirming that these oversampled areas indeed arise in various types of text.

\begin{figure}[t]
  \begin{center}
    \includegraphics[width=0.90\textwidth]{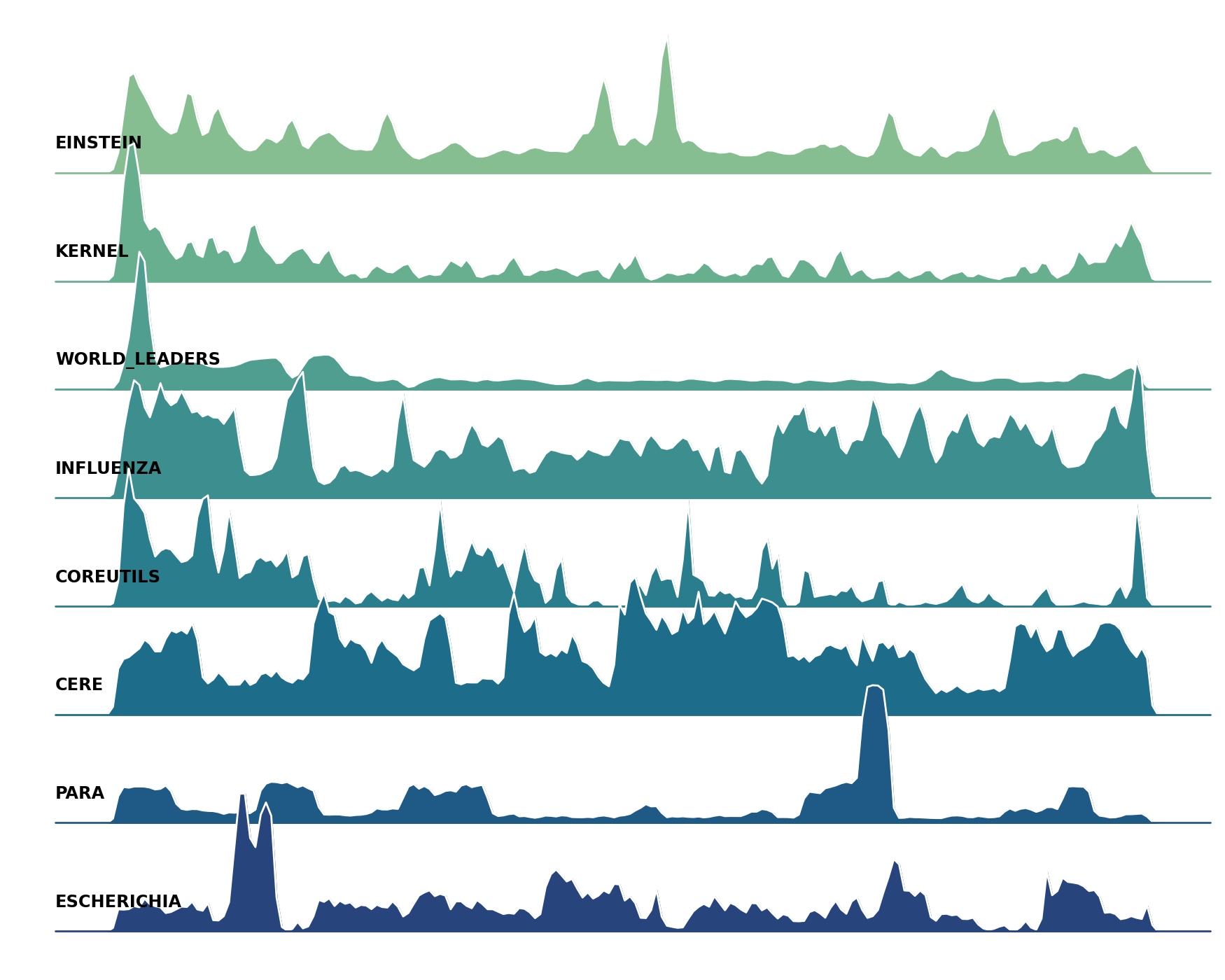}
  \end{center}
  \caption{
    Distribution of $\BWT$-run heads within the Pizza-Chili repetitive texts \cite{PizzaChili}.
    The $x$-axis represents the positions of these run heads along the text.
    The $y$-axis shows the density, which indicates how frequently $\BWT$-run heads appear at different locations.
    The smooth curve is obtained using the statistical method Kernel Density Estimation (KDE).
  }
  \label{fig:samples-densities}
\end{figure}

The existence of those short blocks in the texts is to be expected. Consider, in the particular case of DNA sequences, the site of a single-nucleotide polymorphism, where some genomes have {\tt A} and the others have {\tt G}, in all cases followed by the same string $\alpha$ and preceded by the same symbols $\beta_m \cdots \beta_1$. Those {\tt A}s and {\tt G}s are likely to be intermingled in the $\BWT$ area of the suffixes that start with $\alpha$, but the characters $\beta_1$ will be separated into two $\BWT$ areas: those of the suffixes $\mathtt{A}\alpha$ and $\mathtt{G}\alpha$. The same will happen to the characters $\beta_2$ (separated in the areas of the suffixes $\beta_1\mathtt{A}\alpha$ and $\beta_1\mathtt{G}\alpha$), $\beta_3$, and so on. But for large enough $m$, $\beta_m \cdots \beta_1$ will be unique in the collection, and the suffix array areas of $\beta_m \cdots \beta_1 \mathtt{A} \alpha$ and $\beta_m \cdots \beta_1 \mathtt{G}\alpha$ will be merged again.

The character $\beta_1$ preceding (in the genomes) the first/last (in the $\BWT$) of those {\tt A}s and of those {\tt G}s is quite likely to be the start/end of a run in the $\BWT$, and so are the corresponding characters $\beta_2$, $\beta_3$, and so on, until the two areas merge. It follows that, if a character in the genomes is at a boundary between runs in the $\BWT$, then the character immediately to its left is quite likely to be as well.  In other words, the characters at boundaries between runs in the $\BWT$, will tend to cluster in the genomes.



This section introduces two new indexing schemes for repetitive texts: the \emph{subsampled $\RIdx$} ($\SRIdx$) and the \emph{subsampled $\RCSA$} ($\SRCSA$).
Both are built by combining aspects of existing methods.
The $\SRIdx$ is a hybrid between the $\RIdx$ and the $\RLFMIdx$.
The $\SRCSA$, on the other hand, combines the $\RCSA$ with the $\RLCSA$.

For the sake of clarity, we will employ the designation \emph{sr-indexes} to refer to our two subsampled solutions ($\SRIdx$ and $\SRCSA$), \emph{r-indexes} to mean both $\BWT/\Psi$-runs based indexes ($\RIdx$ and $\RCSA$), and \emph{rl-indexes} to represent the run-length based indexes $\RLFMIdx$ and $\RLCSA$.

Similarly to their corresponding \emph{r}-indexes, {\em sr}-indexes take text position samples at the beginning or end of each run.
Yet, they remove samples from text areas where there are too many of them, ensuring that no three samples in a row are closer together than a certain distance (defined by a parameter $s$).
This approach is a relaxation of the regular text sampling used in the \emph{rl}-indexes, where consecutive samples are separated exactly by $s$ text positions. Unlike \emph{rl}-indexes, some consecutive samples can be very far apart in some areas of the text, but unlike \emph{r}-indexes, \emph{sr}-indexes ensure that they are never too close to each other.
To achieve this goal, the \emph{sr}-indexes face various challenges related to maintaining correctness and efficiency upon removal of samples, in particular ensuring, like \emph{rl}-indexes, that they never require more than $s$ steps to simulate a backward step or computation of an entry of $\SA$.

Our $\SRIdx$ and $\SRCSA$ are based on a similar design, with the primary distinction being the use of the $\fnPhi$ or $\fnIPhi$ function, respectively.
The $\fnPhi$ function employs a predecessor data structure to locate the remaining values in $\SA[\vSPosPat..\vEPosPat-1]$, whereas the $\fnIPhi$ function relies on the successor to compute $\SA[\vSPosPat+1..\vEPosPat]$ (recall that the $\LF$ function of the $\RIdx$ is the inverse of the $\Psi$ function of the $\RCSA$).
To avoid redundant explanations, we will focus our attention on the $\SRIdx$, making pertinent remarks when differences with the $\SRCSA$ are significant and require further explanation.
We will directly present the practical data structures that implement the subsampled indexes.

\subsection{Subsampling}

The $\RIdx$ locating structures are formed by the following components, analogous to those of the $\RCSA$ we described in Section~\ref{subsec:rcsa-practical}.
Our subsampling solutions will later modify them.

\begin{description}
  \item[{$\SamplesL[1..\vnBWTRuns]$}:] an array of $\vnBWTRuns$ sampled text positions, where $\SamplesL[p] = \SA[i]-1$ iff $\BWT[i]$ is the last letter in the $p$th $\BWT$-run.

  \item[{$\MarksF[1..\vsTColl]$}:] a bitvector marking with $1$s the \emph{text} positions of the letters that are the first in a $\BWT$-run.
  That is, if $i = 1$ or $\BWT[i] \not= \BWT[i-1]$, then $\MarksF[\SA[i] - 1] = 1$.
  This allows one to find the rightmost $1$ up to position $j$,
  \begin{equation*}
    \predecessor(\MarksF, j) = \select(\MarksF, \rank(\MarksF, j)),
  \end{equation*}
  with $\select(\MarksF,0) = 0$.

  \item[{$\MapMarksF[1..\vnBWTRuns]$}:] an array mapping each letter marked in $\MarksF$ to the $\BWT$-run preceding the one in which it is located.
  If $\MarksF[j] = 1$ with $j = \SA[i] - 1$, then there exists $p$ such that $\SamplesL[p] = \SA[i - 1] - 1$, because $\BWT[i]$ is the first letter in a $\BWT$-run,\footnote{Note that $i$ cannot belong to the first run (i.e., $i=1$), as then we would be on the suffix $\mathcal{T}[\SA[1]] = \$$.}
  and
  \begin{equation*}
    p = \mapping(\MarksF, j) = \MapMarksF[\rank(\MarksF, j)],
  \end{equation*}
  where $\MapMarksF[0]$ yields the $\BWT$-run preceding to the run of the special symbol \$.
\end{description}

The $\RIdx$ computes the $\fnPhi$ function as
\begin{equation} \label{eq:rindex_phi}
  \fnPhi(j) = \SamplesL[\mapping(\MarksF, j - 1)] + (j - \predecessor(\MarksF, j - 1)).
\end{equation}

The $\SRIdx$ subsampling process removes $\RIdx$ samples in oversampled areas.
Concretely, let $\vSamp_1 < \cdots < \vSamp_{\vnBWTRuns}$ be the text positions of the last letters in $\BWT$-runs, that is, the sorted values in array $\SamplesL$.
For any $1 < i < \vr$, the sample $\vSamp_{i}$ is removed if $\vSamp_{i+1} - \vSamp_{i-1} \le \vSampFac$, where $\vs$ is a parameter.
This condition is tested and applied sequentially for $i = 2,\ldots,\vr-1$.
If, for example, we removed $\vSamp_2$ because $\vSamp_3 - \vSamp_1 \le \vs$, then we next remove $\vSamp_3$ if $\vSamp_4-\vSamp_1 \le \vs$; if we had not removed $\vSamp_2$, then we remove $\vSamp_3$ if $\vSamp_4-\vSamp_2 \le \vs$.
Let us designate $\vt_1,\vt_2,\ldots$ as the sequence of the {\em remaining} samples.

The structures $\MarksF$, $\MapMarksF$, and $\SamplesL$ are constructed exclusively on the remaining subsamples $\vt_i$.
Consequently, the removal of the sample $\SamplesL[p] = \vSamp$ also entails the removal of the $1$ in $\MarksF$ corresponding to the first letter of the $(p+1)$th $\BWT$-run, which is the very instance that Eq.~\eqref{eq:rindex_phi} would have addressed with $\SamplesL[p]$.
In other words, if $i$ is the first position of the $(p+1)$th run and $i-1$ the last of the $p$th run, then if we remove $\SamplesL[p] = \SA[i-1] - 1$, we remove the corresponding $1$ at position $\SA[i] - 1$ in $\MarksF$.
In addition, the corresponding entry of $\MapMarksF$ is also removed.
Finally, note that $\MapMarksF$ must be adapted to point to the corresponding entry of $\SamplesL$, once some entries of the latter are removed.

Subsampling proves to be an effective method for avoiding the excessive space required to store the locating structures.
This is particularly beneficial when the number of $\BWT$-runs $\vnBWTRuns$ is a relatively large value.
In such cases, subsampling can reduce the entries in those data structures from $\Oh(\vnBWTRuns)$ to $\Oh(\min(\vnBWTRuns,\vsTColl/\vSampFac))$ in the worst case (the reduction is much higher in practice because the samples are not uniformly distributed, as seen in Figure~\ref{fig:samples-densities}).

\begin{lemma} \label{lem:space}
The subsampled structures \emph{$\SamplesL$}, \emph{$\MarksF$} and \emph{$\MapMarksF$} use
$\min(\vnBWTRuns, 2 \lceil \vsTColl / (\vSampFac + 1) \rceil) \cdot (2 \log n + \Oh(1))$ bits of space.
\end{lemma}

\begin{proof}
If $x$ is the number of remaining samples, then for each remaining sample array $\SamplesL$ uses $\lceil \log n\rceil$ bits, bitvector $\MarksF$ uses $\log(n/x)+\Oh(1)$ bits \cite{OS07}, and $\MapMarksF$ uses $\lceil \log x \rceil$ bits. The combined size of the three arrays is then $x \cdot (2 \log n + \Oh(1))$ bits. This is the same space as in the implemented $\RIdx$ \cite{GagieNP20:FFS}, with the number of samples reduced from $\vr$ to $x$.

Our subsampling process begins with $\vnBWTRuns$ samples and subsequently removes a subset of them, thus ensuring that the number of samples never exceeds $\vnBWTRuns$.
By construction, any remaining sample $\vt_i$ is guaranteed to satisfy $\vt_{i+1} - \vt_{i-1} > \vSampFac$, so if we cut the text into blocks of length $\vSampFac + 1$, no block can contain more than $2$ samples. Therefore, $x \le \min(\vr,2\lceil \vn/(\vs+1)\rceil)$.
\end{proof}

Our indexes add the following small structure on top of the above ones, so as to mark the removed samples:

\begin{description}
\item[{\rm $\Removed[1..r]$}:] a bitvector telling which of the original samples have been removed.
Specifically, $\Removed[p] = 1$ iff the sample at the end of the $p$th $\BWT$-run was removed.
We can compute any $\rank(\Removed, p)$ in constant time using $\vnBWTRuns + \oh(\vnBWTRuns)$ bits~\cite{Cla96}, as well as $\rank_0(\Removed, p) = p - \rank(\Removed,p)$ (which counts the 0s in $\Removed[1..p]$).
\end{description}


\begin{figure}[t]
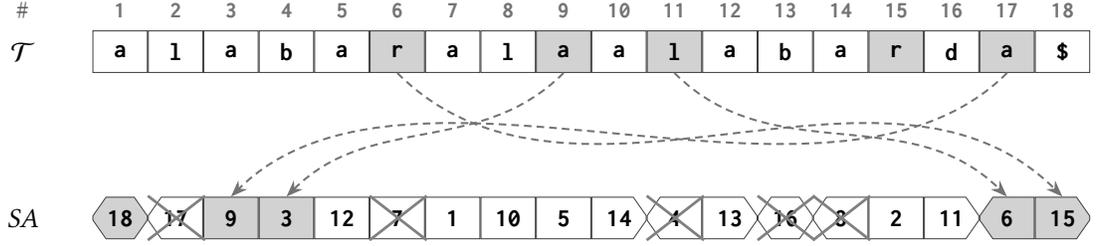

  \begin{center}
    \begin{tikzpicture}[font=\ttfamily \bfseries \small, transform shape]

  \input{tikz_styles}
  \tikzstyle{CellHeader} = [CellHeaderBase]
  \tikzstyle{Cell} = [CellBase]

  \begin{scope}[shift=({0,+3.5ex})]
    \tikzstyle{CellHeader} = [CellHeaderBase, black!65]
    \tikzstyle{Cell} = [CellBase, draw=none, font=\ttfamily \bfseries \footnotesize, black!55]
    \input{base-pos}
  \end{scope}

  \begin{scope}[shift=({0,0})]
    \input{srcsa-scheme-text}
  \end{scope}

  \begin{scope}[shift=({0,-14ex})]
    \input{srcsa-scheme-sa}
    \input{srcsa-scheme-sa-del}
  \end{scope}

  \begin{scope}[max distance=9em, out=225, in=45]
    \draw[ArrowDash] (T-16.south) to [out=225, in=45] (SA-2.north);
    \draw[ArrowDash] (T-8.south) to [out=225, in=45] (SA-3.north);
    \draw[ArrowDash] (T-10.south) to [out=315, in=135] (SA-16.north);
    \draw[ArrowDash] (T-5.south) to [out=315, in=135] (SA-17.north);
  \end{scope}

\end{tikzpicture}
  \end{center}
  \caption{
    Subsampled data structures used for locating mechanism of the $\SRCSA$ using a sampling factor $\vSampFac = 4$ (compare with the full sampling in Figure~\ref{fig:rcsa-scheme}).
    The gray cells in $\vTColl$ are the subsampled elements marked in the bitvector $\MarksL$.
    The gray cells in $\SA$ are the remaining samples stored in the array $\SamplesF$.
    The crossed cells in $\SA$ are the removed samples marked in the bitvector $\Removed$.
    The arrows show the mapping $\MapMarksL$.
  }
  \label{fig:srcsa-scheme}
\end{figure}

\subsubsection*{Construction.}
Once the basic structures of the $\RIdx$ able to perform $\LF$-steps are built, we can create the additional structures $\MarksF$, $\SamplesL$, $\MapMarksF$, and $\Removed$ in two additional virtual backward passes over the text, in $\Oh(\vsTColl \log(\sigma + \vsTColl / \vnBWTRuns))$ time and without extra space. A first pass performs the subsampling, in text order, thereby defining the bitvector $\Removed$. Once we know the number of remaining samples, a second traversal fills the $1$s in $\MarksF$ and the entries in $\MapMarksF$ and $\SamplesL$ for the runs whose sample was not removed.


\subsubsection*{Running on the $\SRCSA$.}
A nearly identical sampling process is used on the $\SRCSA$, using the samples in the $\SamplesF$ array.
Recall that, unlike the $\RIdx$, these samples are the text positions of the first element in each $\Psi$-run.
Given this condition along with the nature of the $\SRCSA$, which relies on the $\Psi$ and $\successor(\MarksL, i)$ functions, we have opted to implement a slight variation.
Specifically, the subsampling mechanism employs a backward iteration over the samples $i = \vnBWTRuns - 1,\ldots, 2$ in $\SamplesF$.
It is straightforward to verify that this change retains the results given above.

The construction process is also similar, using two forward text traversals simulated using the $\Psi$ function. The construction time is $\Oh(n\log B)$.

\subsection{Counting with Toehold} \label{sec:problem1}

The counting algorithm of the practical $\RIdx$ is based on the data structures of a $\RLFMIdx$ variant called {\em sparse $\nameIdx{RLBWT}$} \cite[Thm.~28]{Prezza17:CCT}.
By applying a sparsification strategy primarily to the $\Start$ bitvector, this structure requires only $\vnBWTRuns \cdot ((1 + \epsilon)\log(\vsTColl / \vnBWTRuns) + \log\vsAlph + \Oh(1))$ bits (largely dominated by the described arrays $\Start$ and $\Letter$) for any small constant $\epsilon>0$.
The time complexity for backward search steps and $\LF$-steps is $\Oh((1/\epsilon) \log(\vsAlph + \vsTColl / \vnBWTRuns))$.

To obtain the necessary text position or \emph{toehold} along the backward search, the practical $\RIdx$ maintains the value $\SA[\vEPosPat_{i}]$ along each interval $[\vSPosPat_{i}..\vEPosPat_{i}]$, for $1 \le i \le \vsPat$.
In the non-trivial cases where $\vPat[i] \neq \BWT[\vEPosPat_{i+1}]$, the end of interval is $\vEPosPat_{i} = \LF(j)$, where $j \in [\vSPosPat_{i+1}..\vEPosPat_{i+1}]$ is the largest position with $\BWT[j] = \vPat[i]$. It is easy to see that $j$ must be the end of a $\BWT$ run, in particular of the $p$th run, with $p = \rank_1(\Start, j)$. As we do not know $j$, this run can be computed as $p = \select_c(\Letter, \rank_c(\Letter, \rank_1(\Start, \vEPosPat_{i+1})))$.
Finding $\SA[\vEPosPat_{i}]$ then requires a straightforward lookup process in the $\RIdx$, since it is precomputed and stored in array $\SamplesL$, where it holds $\SA[\vEPosPat_{i}] = \SA[j] - 1 = \SamplesL[p]$.

However, the $\SRIdx$ might have removed the sample $\SamplesL[p] = \SA[\vEPosPat_{i}]$ during its subsampling process, which is indicated by the flag $\Removed[p] = 1$.
When this happens, we use an iterative search, computing $j_k=\LF^k(j)$ for $k=1,2,\ldots$ and $j=\LF^{-1}(\vEPosPat_{i})$, until a remaining sampled $\SA[j_k]$ is found.
This is identified because $j_k$ is the last position in a $\BWT$-run (i.e., $j_k = \vsTColl$ or $\Start[j_k + 1] = 1$) and $\Removed[q] = 0$ for $q = \rank_1(\Start, j_k)$.
When we find such $j_k$, we can compute the final $\SA[\vEPosPat_{i}]$ by adjusting the sample found, $q' = \rank_0(\Removed, q)$, based on the $k$ steps in the search, obtaining $\SA[\vEPosPat_i] = \SamplesL[q'] + k$.

The number of $\LF$-steps in each backward search iteration of the $\SRIdx$ is bounded by the sampling factor $\vSampFac$:
the next lemma shows that for some $k < \vSampFac$ we will find a non-removed sample.

\begin{lemma} \label{lem:distance-s}
If there is a removed sample $\vSamp_j$ such that $\vt_i$ and $\vt_{i+1}$ are remaining samples satisfying $\vt_i < \vSamp_j < \vt_{i+1}$, then $\vt_{i+1} - \vt_i \le \vSampFac$.
\end{lemma}

\begin{proof}
Since our subsampling process removes samples from left to right, by the time we removed $\vSamp_j$, the current sample $\vt_i$ was already the nearest remaining sample to the left of $\vSamp_j$.
If the sample following $\vSamp_j$ was the current $\vt_{i+1}$, then $\vSamp_j$ was removed because $\vt_{i+1}-\vt_i \le s$.
Therefore, the lemma holds in this case.

Otherwise, there were other samples to the right of $\vSamp_j$, say $\vSamp_{j+1}, \vSamp_{j+2}, \ldots, \vSamp_{j+k}$, which were consecutively removed until the current sample $\vt_{i+1}$ was reached.
First, we removed $\vSamp_j$ because $\vSamp_{j+1} - \vt_i \le s$.
Then, for $1 \le l < k$, we removed $\vSamp_{j+l}$  (after having removed $\vSamp_j, \vSamp_{j+1},\ldots,\vSamp_{j+l-1}$) as $\vSamp_{j+l+1}-\vt_i \le s$.
Finally, we removed $\vSamp_{j+k}$ since $\vt_{i+1}-\vt_i \le s$, thus that the lemma holds in this case as well.
\end{proof}

This implies a fixed bound on the search, beginning from the position $j$ of the removed sample $\vSamp = \SA[j] - 1$ and extending to the surrounding remaining samples $\vt_i < \vSamp < \vt_{i+1}$.
It is sufficient to perform $k = \vSamp - \vt_i < \vSampFac$ $\LF$-steps until position $j_k = \LF^{k}(j)$ satisfies $\SA[j_k]-1=\vt_i$, which is stored in $\SamplesL[q']$ and not removed.

If we followed verbatim the modified backward search of the $\RIdx$, finding every $\SA[\vep_i]$, we would perform $\Oh(\vsPat \cdot \vSampFac)$ steps on the $\SRIdx$.
We now reduce this to $\Oh(\vsPat + \vSampFac)$ steps by noting that the only value we require is $\SA[\vep]=\SA[\vep_1]$.
Further, we need to know $\SA[\vep_{i+1}]$ to compute $\SA[\vep_i]$ only in the easy case where $\BWT[\vep_{i+1}]=\vP[i]$ and so $\SA[\vep_i]=\SA[\vep_{i+1}]-1$.
Otherwise, the value $\SA[\vep_i]$ is computed afresh.

We then proceed as follows.
We do not compute any value $\SA[\vep_i]$ during the backward search; we only remember the last (i.e., smallest) value $i'$ of $i$ where the computation was not easy, that is, where $\BWT[\vep_{i'+1}] \not= \vP[i']$.
Then, $\SA[\vep_1] = \SA[\vep_{i'}]-(i'-1)$ and we need to apply the procedure described above only once: we compute $\SA[j]$, where $j$ is the largest position in $[\vsp_{i'+1}..\vep_{i'+1}]$ where $\BWT[j]=P[i']$, and then $\SA[\vep_{i'}]=\SA[j]-1$.

\begin{algorithm}[t]
  \IO{Query pattern $\vPat[1..\vsPat]$.}
  {Range $\langle \vsp,\vep \rangle$ on $\SA$ for $\vPat$;
  Value $\SA[\vep]$.}

  \begin{multicols}{2}
    \SetKwFunction{fnCount}{count}
\SetKwFunction{fnToehold}{findToehold}

\Function{\fnCount{$\vPat[1..\vsPat]$}}{
  $\langle \vsp, \vep \rangle \leftarrow \langle 1, \vsTColl \rangle$\;

  $\langle i_{v}, p_{v} \rangle \leftarrow \langle -1, -1 \rangle$\;

  \For{$i \leftarrow \vsPat$ \KwDownTo $1$}{
    $c \leftarrow \vPat[i]$\;
    $p \leftarrow \rank_{1}(\Start, \vep)$\;

    \If{$c \ne \Letter[p]$} {
      $\langle i_{v}, p_{v} \rangle \leftarrow \langle i, p \rangle$\;
    }


    $\vsp \leftarrow C[c] + \rank_{c}(\BWT, \vsp{-}1) {+} 1$\;
    $\vep \leftarrow C[c] + \rank_{c}(\BWT, \vep)$\;
    \If{$\vsp > \vep$} {\Return{``$\vPat$ is not in $\vTColl$''}\;}
  }


  $v \leftarrow \fnToehold(i_{v}, p_{v})$\;

  \Return{$\langle \vsp, \vep, v \rangle$}\;
}

\BlankLine

\Function{\fnToehold{$i_{v}$, $p_{v}$}}{
  \If{$i_{v} = -1$}{
    \tcp{$\SA[\vsTColl]$ is stored}
    \Return{$\SA[\vsTColl] - \vsPat$}\;
  }


  $c \leftarrow \vPat[i_{v}]$\;
  $q \leftarrow \select_{c}(\Letter, \rank_{c}(\Letter, p_{v}))$\; 
  $j \leftarrow \select_{1}(\Start, q + 1) - 1$\;
  $k \leftarrow 0$\;
  \While{$(j < \vsTColl$ \KwAnd \Start$[j + 1] = 0)$ \KwOr $\Removed[q] = 1$}
  { $j \leftarrow \LF(j)$ \\
    $q \leftarrow \rank_1(\Start, j)$ \\
    $k \leftarrow k + 1$
  }

  \Return{$\SamplesL[\rank_{0}(\Removed, q)] {+} k {-} (i_{v} {-} 1)$}\;
}

  \end{multicols}
  \BlankLine

  \caption{Counting pattern occurrences and finding value $\SA[\vep]$ with the $\SRIdx$.}
  \label{alg:srindex_count_toehold}
\end{algorithm}

Algorithm~\ref{alg:srindex_count_toehold} gives the complete pseudocode that counts while finding a toehold.
Note that, if $\vP$ does not occur in $\vT$ (i.e., $\occ=0$) we realize this after the
$\Oh(m)$ backward steps because some $\vsp_i > \vep_i$, and thus we do not spend the $\Oh(s)$ extra steps.

\paragraph{Running on the $\SRCSA$.}
Although it is possible to simulate the $\LF$-steps via the $\Psi$ function, a more intuitive and efficient approach is to search the text forwards for a sampled position, using $\Psi$ directly.
Consequently, starting from the eliminated sample $\vSamp$, the $\SRCSA$ searches for the next remaining sample, $\vt_{i+1}$, rather than for the preceding one, $\vt_{i}$.
Lemma~\ref{lem:distance-s} also shows that the subsampling factor $\vSampFac$ bounds the number of steps in this case.
Therefore, for some $k < \vSampFac$, sample $\vt_{i+1} = \SA[\Psi^{k}(j)]$ is a $\Psi$-run head stored in $\SamplesF$, with $\vSamp = \SA[j]$ (note that we use $\SamplesF$ instead of $\SamplesL$ because on the $\SRCSA$ we obtain $\SA[sp]$, not $\SA[ep]$).


\subsection{Locating from Toehold} \label{sec:problem2}

We now focus on the problem of finding $\SA[j-1]$ from $i=\SA[j]-1$. If we just apply the $\RIdx$ procedure based on $\fnPhi$, we will end up at an incorrect predecessor if the correct one has been removed during subsampling. To circumvent this problem, the $\SRIdx$ will start with a procedure similar to the one used to identify the toehold, that is, iteratively searching for a remaining sample. We will show that, when this procedure fails, we can safely use the $\fnPhi$ function like the $\RIdx$.

In this context, we first calculate $j_k' = \LF^{k}(j-1)$ for $k = 0, \ldots, \vSampFac - 1$.
For each of those $j_k'$ we verify whether it is the last symbol of its run (i.e., $j_k' = \vsTColl$ or $\Start[j_k' + 1] = 1$), and the sample corresponding to this run has not been removed (i.e., $\Removed[q] = 0$, with $q = \rank_1(\Start, j_k')$).
If these conditions are met, then it can be immediately derived that $\SA[j_k'] = \SamplesL[q'] + 1$, where $q' = \rank_0(\Removed, q)$.
Consequently, we can also obtain $\SA[j-1] = \SA[j_k'] + k$.

Unlike in Section~\ref{sec:problem1}, the symbol $\BWT[j - 1]$ is not necessarily an end of run.
Therefore, there is no guarantee that a solution will be found for some $0 \le k < \vSampFac$.
However, the following property shows that, if there were some ends of runs $j_k'$, it is not possible that all were removed from $\SamplesL$.

\begin{lemma} \label{lem:sri-no-removed-samples}
If there are no remaining samples in $\SA[j-1]-\vSampFac,\ldots,\SA[j-1]-1$, then no sample was removed between $\SA[j-1]-1$ and its preceding remaining sample.
\end{lemma}

\begin{proof}
Let $\vt_i < \SA[j-1]-1 < \vt_{i+1}$ be the samples surrounding $\SA[j-1]-1$, so the remaining sample preceding $\SA[j-1]-1$ is $\vt_i$.
Since $\vt_i < \SA[j-1]-s$, it follows that $\vt_{i+1}-\vt_i > s$ and thus, by Lemma~\ref{lem:distance-s}, no samples were removed between $\vt_i$ and $\vt_{i+1}$.
\end{proof}

This means that, if the above process fails to yield an answer, it is possible to employ Eq.~\eqref{eq:rindex_phi} directly, as proved next.

\begin{lemma} \label{lem:sri-no-removed-marks}
If there are no remaining samples in $\SA[j-1]-\vSampFac, \ldots, \SA[j-1]-1$, then subsampling removed no $1$s in \emph{$\MarksF$} between positions $i=\SA[j]-1$ and \emph{$\predecessor(\MarksF, i)$}.
\end{lemma}

\begin{proof}
Let $\vt_i < \SA[j-1]-1 < \vt_{i+1}$ be the samples surrounding $\SA[j-1]-1$, and
$k=\SA[j-1]-1-\vt_i$.
Lemma~\ref{lem:sri-no-removed-samples} implies that no sample existed between $\SA[j-1]-1$ and $\SA[j-1]-k=t_i+1$, and there exists one at $t_i$.
Consequently, no $1$ existed in $\MarksF$ between positions $\SA[j]-1$ and $\SA[j]-k$ (both included), and there exists one in $\SA[j]-1-k$.
Indeed, $\predecessor(\MarksF,i)=\SA[j]-1-k$.
\end{proof}

An additional optimization, which does not alter the worst-case complexity but enhances performance in practice, is to reuse work across successive occurrences.
Let $\BWT[\vSPosMRun..\vEPosMRun]$ be a maximal run inside $\BWT[\vsp..\vep]$.
For every $\vSPosMRun \le j \le \vEPosMRun$, the first $\LF$-step will result in $\LF(j) = \LF(\vSPosMRun)+(j-\vSPosMRun)$.
Therefore, the entire sequence of values $\LF(j)$ can be computed through a single iteration of the $\LF$ function.

Consequently, rather than locating $\SA[\vsp], \ldots, \SA[\vep]$ one by one, we first report $\SA[\vep]$ (which has been previously identified), and then partition $\BWT[\vsp..\vep-1]$ into maximal runs using bitvector $\Start$. We will traverse those maximal runs $\BWT[sm..\vEPosMRun]$, from largest to smallest $sm$, with the invariant that $\SA[em+1]$ is known. We first check that the run end $\BWT[\vEPosMRun]$ is sampled, in which case we report its position and decrement $em$ (note that the offset $k$ must be added to all the results reported at level $k$ of the recursion). We then continue recursively with 
$\SA[\LF(\vSPosMRun)..\LF(\vSPosMRun)+(\vEPosMRun-\vSPosMRun)]$.
By Lemma~\ref{lem:distance-s}, every non-removed sample found during the traversal has been duly reported prior to level $k=\vSampFac$.
Upon reaching the final recursion level ($k = s$), we use Eq.~\eqref{eq:rindex_phi} to obtain $\SA[\vEPosMRun], \ldots, \SA[\vSPosMRun]$ consecutively from $\SA[\vEPosMRun+1]$.
Algorithm~\ref{alg:srindex_locate} gives the complete procedure to locate the occurrences.

\begin{algorithm}[t]
  \IO{Query pattern $\vPat[1..\vsPat]$.}
  {Occurrences of $\vPat$: $V[1..\vnOcc] = \SA[\vsp..\vep]$.}

  \begin{multicols}{2}
    \SetKwFunction{fnLocate}{locate}
\SetKwFunction{fnCount}{count}
\SetKwFunction{fnLocateRec}{locateRec}
\SetKwFunction{fnLocateBase}{locateBase}
\SetKwFunction{fnCheckSample}{checkSample}

\Function{\fnLocate{$\vPat[1..\vsPat]$}}{
  $\langle \vsp, \vep, v \rangle \leftarrow \fnCount(\vPat)$\;

  $V[\vep - \vsp + 1] \leftarrow v$\;

  \If{$\vsp < \vep$} {
    \fnLocateRec($\vsp$, $\vep - 1$, $0$)\;
  }

  \Return{$V$}\;
}

\BlankLine

\Function{\fnCheckSample{$\vem$, $k$}}{
  $p \leftarrow \rank_1(\Start, \vem)$ \;
  \If{$\Removed[p] = 1$} {
    \Return{$\vem$}
  }


  $V[\vem - \vsp + 1] \leftarrow$\;
  \ \ \ \ \ \ \ \ \ \ $\SamplesL[\rank_0(\Removed, p)] + 1 + k$ \;
  \Return{$\vem - 1$} \;
}

\BlankLine

\Function{\fnLocateRec{$\vSPosMRun$, $\vEPosMRun$, $k$}} {
  \If{$k = \vSampFac$} {
    \For{$i \leftarrow \vem$ \KwDownTo $\vsm$}{
      $V[i - \vsp + 1] \leftarrow \fnPhi(V[i - \vsp + 2])$\;
    }
    \Return{} \;
  }

  $p \leftarrow \rank_1(\Start, \vem)$ \;
  $q \leftarrow \rank_1(\Start, \vsm)$ \;
  \For{$j \leftarrow p$ \KwDownTo $q$}{
    \If{$j < p$ \KwOr $\Start[\vem + 1] = 1$} {
      $em \leftarrow $\fnCheckSample($\vem$, $k$) \;
    }

    $im \leftarrow \max(\select_1(\Start, j), \vsm)$ \;

    \If{$im > \vem$} { \Continue }

    \fnLocateRec($\LF(im)$, $\LF(\vem)$, $k+1$) \; \label{line-alg:srindex_locate:call-locate_rec}
    $\vem \leftarrow im - 1$ \;
  }

}

  \end{multicols}
  \BlankLine

  \caption{Locating pattern occurrences with the $\SRIdx$.}
  \label{alg:srindex_locate}
\end{algorithm}

\paragraph{Running on the $\SRCSA$.}
The aforementioned results are directly applicable to the $\RCSA$ to obtain the $\SRCSA$.
On this structure, the problem is to compute the value $\SA[j+1]$ based on the previously determined value $i = \SA[j]$.
It is a straightforward exercise to prove the symmetric version of Lemmas~\ref{lem:sri-no-removed-samples} and \ref{lem:sri-no-removed-marks} needed for the $\SRCSA$.
The $\SRCSA$ probes the range $\SA[j+1], \ldots, \SA[j+1]+\vSampFac-1$ using $\Psi^k(j+1)$ for $0 \le k < s$, looking for a non-removed sample.
If this fails, it makes use of the $\fnIPhi$ function of the $\RCSA$.
Finally, a similar recursive process can be employed to avoid computing $\SA[\vsp], \ldots, \SA[\vep]$ individually.
In this case, the procedure uses the maximal runs within $\Psi[\vsp..\vep]$.

\subsection{The basic \emph{sr}-indexes, \texorpdfstring{$\SRIdx_0$ and $\SRCSA_0$}{\emph{sr}-index0 and \emph{sr}-CSA0}}

We have just described our most space-efficient index, which we call $\SRIdx_0$.
Its space and time complexity is established in the next theorem.

\begin{theorem} \label{thm:srindex}
The $\SRIdx_0$ uses $\vnBWTRuns \cdot ((1+\epsilon)\log(\vsTColl / \vnBWTRuns) + \log\sigma + \Oh(1)) + \min(\vnBWTRuns, 2\lceil \vsTColl / (\vSampFac + 1)\rceil) \cdot 2\log \vsTColl$ bits of space, for any constant $\epsilon > 0$, and finds all the $\vnOcc$ occurrences of $\vPat[1..\vsPat]$ in $\vTColl$ in time $\Oh((1/\epsilon)(\vsPat + \vSampFac \cdot \vnOcc) \log(\sigma + \vsTColl / \vnBWTRuns))$.
\end{theorem}

\begin{proof}
The space is the sum of the counting structures of the $\RIdx$ and our modified locating structures, according to Lemma~\ref{lem:space}.
The space of bitvector $\Removed$ is $\Oh(\vnBWTRuns)$ bits, which is accounted for in the formula.

As for the time, we have seen that the modified backward search requires $\Oh(\vsPat)$ steps if $\vnOcc=0$ and $\Oh(\vsPat + \vSampFac)$ otherwise (Section~\ref{sec:problem1}).
Each occurrence is then located in $\Oh(\vSampFac)$ steps (Section~\ref{sec:problem2}).
In total, we complete the search with $\Oh(\vsPat + \vSampFac \cdot \vnOcc)$ steps.

Each step involves $\Oh((1/\epsilon)\log(\sigma + \vsTColl / \vnBWTRuns))$ time in the basic $\RIdx$ implementation, including Eq.~\eqref{eq:rindex_phi}.
Our index includes additional $\rank$s on $\Start$ and other constant-time operations, which are all in $\Oh(\log(\vsTColl / \vnBWTRuns))$.
Since $\MarksF$ now has $\Oh(\min(\vnBWTRuns, \vsTColl / \vSampFac))$ $1$s, however, operation $\rank_1$ on it takes time $\Oh(\log (\vsTColl / \min(\vnBWTRuns, \vsTColl / \vSampFac))) =
\Oh(\log\max(\vsTColl / \vnBWTRuns, \vSampFac)) = \Oh(\log(\vsTColl / \vnBWTRuns + \vSampFac))$.
Yet, this $\rank$ is computed only once per occurrence reported, when using Eq.~\eqref{eq:rindex_phi}, so the total time per occurrence is still $\Oh(\log(\vsTColl / \vnBWTRuns + \vSampFac) + \vSampFac \cdot \log(\sigma + \vsTColl / \vnBWTRuns)) = \Oh(\vSampFac \cdot \log(\sigma + \vsTColl / \vnBWTRuns))$.
\end{proof}

Note that, in asymptotic terms, the $\SRIdx$ is never worse than the $\RLFMIdx$ with the same value of $\vSampFac$ and, with $\vSampFac=1$, it boils down to the $\RIdx$.
Using predecessor data structures of the same asymptotic space of our lighter sparse bitvectors, the logarithmic times can be reduced to loglogarithmic~\cite{GagieNP20:FFS}, but our focus is on low practical space usage.

A similar analysis, combined with Theorem~\ref{thm:rcsa_practical}, yields the analogous result for the $\SRCSA$.

\begin{theorem} \label{thm:srcsa}
The $\SRCSA_0$ uses
\begin{equation*}
  \vnBWTRuns \cdot \big((2\log(\vsTColl / \vnBWTRuns) + \log\sigma)(1 + \epsilon) + \Oh((\log\vsTColl) / B) + \Oh(1)\big) + \min(\vnBWTRuns, 2\lceil \vsTColl / (\vSampFac + 1)\rceil) \cdot 2\log \vsTColl
\end{equation*}
bits of space, for any constant $\epsilon, B > 0$, and finds all the $\vnOcc$ occurrences of $\vPat[1..\vsPat]$ in $\vTColl$ in time $\Oh(\vsPat (\log\vnBWTRuns + B) + \vnOcc (\vSampFac (\log\vnBWTRuns + B) + \log\vsTColl))$.
\end{theorem}
\begin{proof}
We carry out $\Oh(m+s\cdot occ)$ steps, each of which computes $\Psi$ at cost $\Oh(\log r + B)$. This yields total time $\Oh((\vsPat + \vSampFac \cdot \vnOcc)(\log\vnBWTRuns + B) + \vnOcc \log(\vsTColl / \vnBWTRuns + \vSampFac))$, which is simplified to the one given once we put together all the terms that are multiplied by $\vnOcc$.
\end{proof}

It should be noted that Theorems~\ref{thm:srindex} and \ref{thm:srcsa} can be obtained by simply choosing the smallest between the $\RIdx$ and the $\RLFMIdx$, or $\RCSA$ and $\RLCSA$, respectively.
In practice, however, the \emph{sr}-indexes perform significantly better than both extremes, providing a smooth transition that retains sparsely indexed areas of $\vTColl$ while removing redundancy in oversampled areas.
This will be demonstrated in Section~\ref{sec:result}.

\subsection{Faster and larger \emph{sr}-indexes, \texorpdfstring{$\SRIdx_1$ and $\SRCSA_1$}{\emph{sr}-index1 and \emph{sr}-CSA1}}

The $\SRIdx_0$ and the $\SRCSA_0$ guarantee locating time proportional to $\vSampFac$.
To do this, however, they perform up to $\vSampFac$ $\LF$-steps or $\Psi$-steps to locate {\em every} occurrence, even when this turns out to be useless.
The $\SRIdx_1$ variant adds a new small component to speed up some cases:

\begin{description}
\item[{\rm $\ValidF$}:] a bitvector storing one bit per (remaining) mark in text
order, so that $\ValidF[q]=0$ iff there were removed values between the $q$th
and the $(q+1)$th $1$s of $\MarksF$.
\end{description}

With this bitvector, if we have $i=\SA[j]-1$ and $\ValidF[\rank_1(\MarksF, i)]=1$, we know that there were no removed values between $i$ and $\predecessor(\MarksF, i)$ (even if they are  less than $\vSampFac$ positions apart).
In this case we can skip the computation of $\LF^k(j-1)$ of $\SRIdx_0$, and directly use Eq.~\eqref{eq:rindex_phi}.
Otherwise, we must proceed exactly as in $\SRIdx_0$ (where it is still possible that we compute all the $\LF$-steps unnecessarily).
More precisely, this can be tested for every value between $\vSPosMRun$ and $\vEPosMRun$ so as to report some further cells before recursing on the remaining ones, in line~\ref{line-alg:srindex_locate:call-locate_rec} of Algorithm~\ref{alg:srindex_locate}.

The $\SRCSA_1$ index employs the analogous bitvector $\ValidL$ to ascertain whether values between $i$ and $\successor(\MarksL, i)$ have been removed.
In this instance, $\ValidL[q] = 0$ indicates that one or more elements were removed between the $(q-1)$th and $q$th $1$s of $\MarksL$.

The space and worst-case complexities of Theorems~\ref{thm:srindex} and~\ref{thm:srcsa} are preserved in $\SRIdx_1$ and $\SRCSA_1$.

\subsection{Even faster and larger, \texorpdfstring{$\SRIdx_2$ and $\SRCSA_2$}{\emph{sr}-index2 and \emph{sr}-CSA2}}

Our second variants, $\SRIdx_2$ and $\SRCSA_2$, add a second and significantly larger structure:

\begin{description}
\item[{\rm $\ValidAreaF$}:] an array whose cells are associated with the $0$s in $\ValidF$.
If $\ValidF[q]=0$, then $d = \ValidAreaF[\rank_0(\Valid,q)]$ is the distance from the $q$th $1$ in $\MarksF$ to the next removed value.
Each entry in $\ValidAreaF$ requires $\lceil \log \vSampFac \rceil$ bits, because removed samples must be at distance less than $\vSampFac$ from their preceding sample, by Lemma~\ref{lem:distance-s}.
\end{description}

If $\ValidF[\rank_1(\MarksF,i)]=0$, then there was a removed sample at $\predecessor(\MarksF,i)+d$, but not before.
So, if $i < \predecessor(\MarksF,i)+d$, we can still use Eq.~\eqref{eq:rindex_phi}; otherwise we must compute the $\LF$-steps $\LF^k(j-1)$ and we are guaranteed to succeed in less than $\vSampFac$ steps.
This improves performance considerably in practice, though the worst-case time complexity stays as in Theorem~\ref{thm:srindex} (and~\ref{thm:srcsa}) and the space increases by at most $\vnBWTRuns \log \vSampFac$ bits.

The $\SRCSA_2$ is based on the same fundamental concept but employs the array $\ValidAreaL$.
Let $\ValidL[q]=0$ and $\ValidAreaL[\rank_0(\ValidL, q)] = d$, for any text position $i$ such that $q = \rank_1(\MarksL, \successor(\MarksL, i))$:
If $i > \successor(\MarksL, i)-d$, the $\fnIPhi$ function (Eq.~\ref{eq:rcsa_iphi}) can be employed;
otherwise, up to $\vSampFac$ $\Psi$-steps must be executed from the corresponding $\SA$ position $j+1$ .

\section{Experimental Results} \label{sec:result}

In this section, we evaluate the performance of the proposed indexing schemes: $\RCSA$, $\SRCSA$, and $\SRIdx$.
These algorithms are implemented using C++17 and the SDSL library\footnote{Available at \url{https://github.com/simongog/sdsl-lite}.}.
The code is publicly available on GitHub (\url{https://github.com/duscob/sr-index}).

We compare our proposed solutions against existing implementations of the $\RIdx$, $\RLCSA$, and other popular indexing methods designed for repetitive text collections.
We aim to assess the effectiveness of the $\SRCSA$ and $\SRIdx$ approaches, particularly in terms of space usage and search speed, across various datasets.

The tests were conducted on a computer with two Intel Xeon processors (Silver 4110 at $2.10$ GHz) and $736$ GB of RAM\@.
The operating system was Debian Linux (version \texttt{5.10.0-0.deb10.16-amd64}).
We compiled the code with the highest optimization settings and disabled multithreading for consistency.

To ensure reliable results, we measured the average user time needed to perform several searches for random patterns of lengths 10, 20, and 30 characters in different text collections.
We report space usage in bits per symbol (bps) and search times in microseconds per occurrence ($\mu$s/occ).
We note that some indexing methods might not be suitable for all text collections or might require excessive space or time to build.
Those methods are excluded from the corresponding graphs in the section of results.

\subsection{Tested indexes}

We include the following indexes in our benchmark; their space decrease as $s$ grows. We chose the parameter range so as to cover the interesting space-time tradeoffs. In particular, larger values of $s$ for the \emph{sr}-indexes only increase the time without significantly reducing the space any further.

\begin{description}
\item[\textsf{\emph{r}-csa}:] Our index implementation, using block size $B=64$ (also for \SRCSA).

\item[\textsf{\emph{sr}-csa}:] Our index, including the three variants, with sampling values $s=4, 8, 16, 32, 64$.

\item[$\SRIdx$:] Our index, including the three variants, with sampling values $s=4, 8, 16, 32, 64$.

\item[$\RIdx$:] The $\RIdx$ implementation we build on.\footnote{From \url{https://github.com/nicolaprezza/r-index}.}

\item[\textsf{rlcsa}:] An implementation of the run-length CSA~\cite{MakinenNSV10:SRH}, which outperforms the actual $\RLFMIdx$ implementation.\footnote{From \url{https://github.com/adamnovak/rlcsa}.}
We use text sampling values $s=\vsTColl / \vnBWTRuns \cdot f / 8$, with $f=8, 10, 12, 14, 16$.

\item[\textsf{csa}:] An implementation of the CSA~\cite{Sad03}, which outperforms in practice the $\FMIdx$~\cite{FerraginaM05:ICT,FerraginaMMN07:CRS}.
This index, obtained from SDSL, acts as a control baseline that is not designed for repetitive collections. 
We use text sampling parameter $s=16, 32, 64, 128$.

\item[$\GIdx$:] The best grammar-based index implementation we are aware of~\cite{CNP21}.\footnote{From \url{https://github.com/apachecom/grammar\_improved\_index}.}
We use Patricia tree sampling values $s=4, 16, 64$.

\item[$\LZIdx$ and $\LZEndIdx$:] Two variants of the Lempel--Ziv based index~\cite{KreftN13:CIR}.\footnote{From \url{https://github.com/migumar2/uiHRDC}.}

\item[$\HybridIdx$:] A hybrid between a Lempel--Ziv and a $\BWT$-based index~\cite{FKP18}.\footnote{From \url{https://github.com/hferrada/HydridSelfIndex}.}
We build it with parameters $M = 8,16$, which are the best for this case. 
\end{description}

\subsection{Collections}

We benchmark various repetitive text collections; Table~\ref{tab:collections} gives some basic measures on them.

\begin{description}
\item[PizzaChili:] A generic collection of real-life texts of various sorts and repetitiveness levels, which we use to obtain a general idea of how the indexes compare.
We use 4 collections of microorganism genomes (\textsf{influenza}, \textsf{cere}, \textsf{para}, and \textsf{escherichia}) and 4 versioned document collections (the English version of \textsf{einstein}, \textsf{kernel}, \textsf{worldleaders}, \textsf{coreutils}).\footnote{From \url{http://pizzachili.dcc.uchile.cl/repcorpus/real}.}

\item[Synthetic DNA:] A 100KB DNA text from PizzaChili, replicated $1{,}000$ times and each copied symbol mutated with a probability from $0.001$ (\textsf{DNA-001}, analogous to human assembled genomes) to $0.03$ (\textsf{DNA-030}, analogous to sequence reads).
We use this collection to study how the indexes evolve as repetitiveness decreases.

\item[Real DNA:] Some real DNA collections to study the performance on more massive data:
\begin{description}

\item[\textsf{Chr19}:] Human assembled genome collections of about 55 billion base pairs, concretely the set of 1{,}000 chromosome 19 genomes taken from the 1000 Genomes Project~\cite{1000genomes}.
\item[\textsf{Salmonella}:] Bacterial assembled genome collections of about 70 billion base pairs, concretely the set of 14{,}609 genomes from the GenomeTrakr project~\cite{stevens2017public}.

\end{description}
\end{description}

\begin{table}[t]
  \begin{center}
    \begin{tabular}{l|r|r || l|r|r}
      Collection & Size & $n/r$ & Collection & Size & $n/r$ \\
      \hline
      \textsf{influenza} & 147.6 & 51.2 &
      \textsf{DNA-001} & 100.0 & 142.4 \\
      \textsf{cere} & 439.9 & 39.9 &
      \textsf{DNA-003} & 100.0 & 58.3 \\
      \textsf{para} & 409.4 & 27.4 &
      \textsf{DNA-010} & 100.0 & 26.0 \\
      \textsf{escherichia} & 107.5 & 7.5 &
      \textsf{DNA-030} & 100.0 & 11.6 \\
      \hline
      \textsf{einstein} & 447.7 & 1{,}611.2 &
      \textsf{Chr19} & 56{,}386.1 & 1{,}287.4 \\
      \textsf{kernel} & 238.0 & 92.4 &
      \textsf{Salmonella} & 70{,}242.6 & 47.0 \\
      \textsf{worldleaders} & 44.7 & 81.9 &
      & & \\
      \textsf{coreutils} & 195.8 & 43.8 &
      & & \\
    \end{tabular}
  \end{center}
  \caption{Basic characteristics of the repetitive texts used in our benchmark. Size is given in MB.}
  \label{tab:collections}
\end{table}

For each dataset, we randomly selected 500 patterns of three different sizes (10, 20, and 30 characters), mixed in a single set.
To ensure the stability and reliability of our results, we identified and removed outliers from the collected patterns using the interquartile range (IQR) method.
This process resulted in a final set of 404 to 465 patterns per dataset.

\subsection{Results}\label{subsec:results}

To simplify our analysis and later comparisons between the indexing methods, we first investigated the effectiveness of different variants within our proposed $\SRIdx$ and $\SRCSA$ methods.

\begin{figure}[t]
  \begin{center}
    \includegraphics[width=0.45\textwidth]{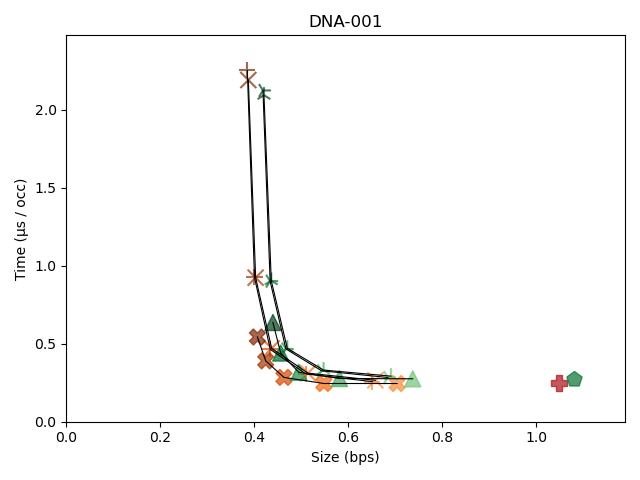}
    \includegraphics[width=0.45\textwidth]{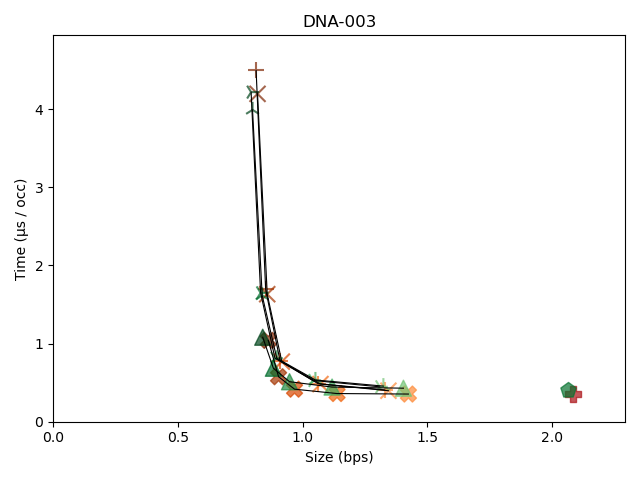}
    \includegraphics[width=0.45\textwidth]{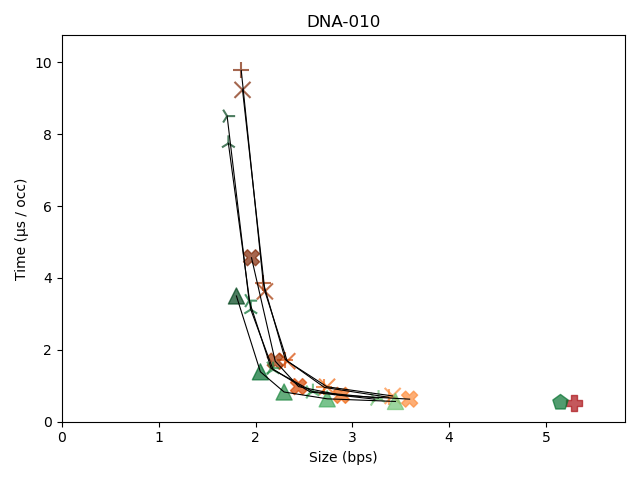}
    \includegraphics[width=0.45\textwidth]{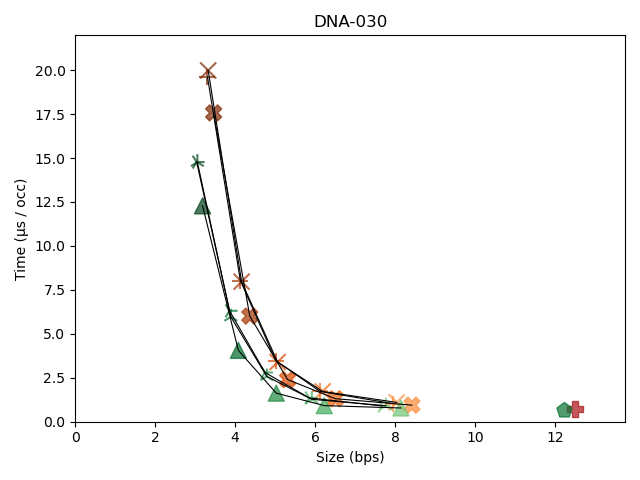}

    \includegraphics[width=0.95\textwidth]{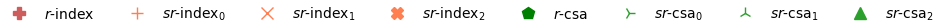}
  \end{center}
  \caption{
    Space-time tradeoffs of $\SRIdx$ and $\SRCSA$ variants on the synthetic DNA collections.
  }
  \label{fig:exp-dna-synthetic-sri}
\end{figure}

Our initial experiments focus on \nameColl{synthetic DNA} datasets; the results are shown in Figure~\ref{fig:exp-dna-synthetic-sri}.
The largest variant stands out as the best choice for both $\SRIdx$ and $\SRCSA$ in terms of balancing space efficiency and search speed.
When the sampling factor $\vSampFac$ is small, all three variants within each method exhibit similar performance.
However, as we increase the parameter $\vSampFac$, the number of samples decreases at a regular rate, but a significant difference in locating speed emerges:
$\nameIdx{\emph{sr}-index}_{2}$ and $\nameIdx{\emph{sr}-csa}_{2}$ demonstrate a substantial improvement in search time compared to the other two, while the required space remains comparable across all variants.

These findings show that the extra data we associate with samples in the third variant has a minor impact in space.
Yet, this additional sample validity information plays a crucial role in how fast the index can locate pattern occurrences.
Given these results, we will simply refer to the third variants as $\SRIdx$ and $\SRCSA$, and use them for the following comparisons.

We also note that, while the $\RIdx$ and the $\RCSA$ are almost indistinguishable, there is some difference in the subsampled variants: the $\SRIdx$ performs better with higher repetitiveness and the $\SRCSA$ stands out with lower repetitiveness, meeting at a mutation probability of 0.003.

\begin{figure}[t]
  \begin{center}
    \includegraphics[width=0.45\textwidth]{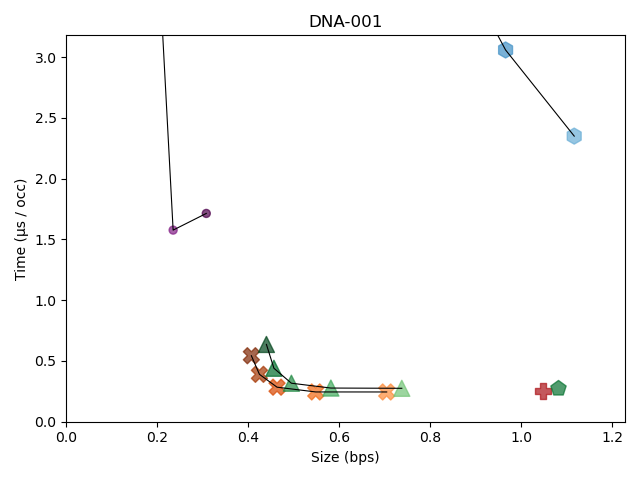}
    \includegraphics[width=0.45\textwidth]{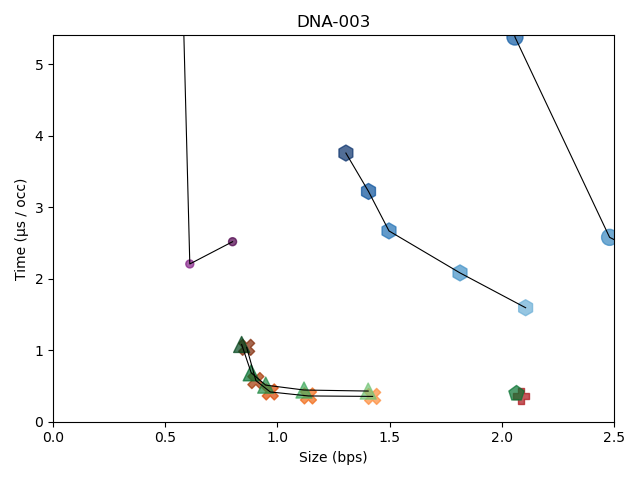}
    \includegraphics[width=0.45\textwidth]{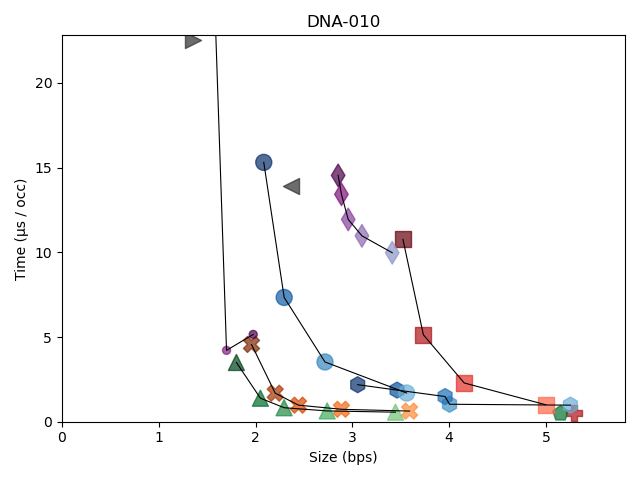}
    \includegraphics[width=0.45\textwidth]{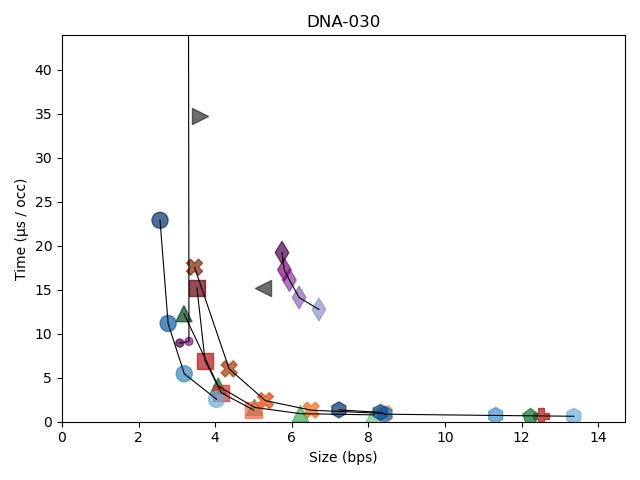}

    \includegraphics[width=0.75\textwidth]{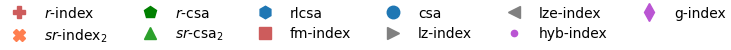}
  \end{center}
  \caption{Space-time tradeoffs on the synthetic DNA collections.}
  \label{fig:exp-dna-synthetic}
\end{figure}

Figure~\ref{fig:exp-dna-synthetic} includes the other state-of-the-art solutions in the comparison.
It can be seen that the $\RIdx$ and our $\RCSA$ are significantly faster than the others on highly repetitive scenarios. They are only matched by the $\FMIdx$ and the $\RLCSA$, which are not designed for repetitive texts, when the mutation probability reaches 0.01 (and the average run length $n/r$ approaches 25).

The dominant indexes, however, are our main contributions: the $\SRIdx$ and the $\SRCSA$. The figures show that they can be almost as fast as the $\RIdx$ and $\RCSA$, while using less than half their space. They sharply dominate the space-time tradeoff map, even with mutation probability as high as 0.01, sweeping out all previous solutions based on the BWT, on $\Psi$, on grammars, and on Lempel-Ziv. The only alternative that stays in the Pareto curve is the \HybridIdx, which in the most repetitive texts can use up to half the space, yet at the price of being an order of magnitude slower. Only when the mutation rate reaches 0.03 and the average run length approaches 10, our \emph{sr}-indexes finally yield to the $\CSA$.

Therefore, as promised, we are able to remove a significant degree of redundancy in the {\em r}-indexes without sacrificing their outstanding time performance.

\begin{figure}[t]
  \begin{center}
    \includegraphics[width=0.45\textwidth]{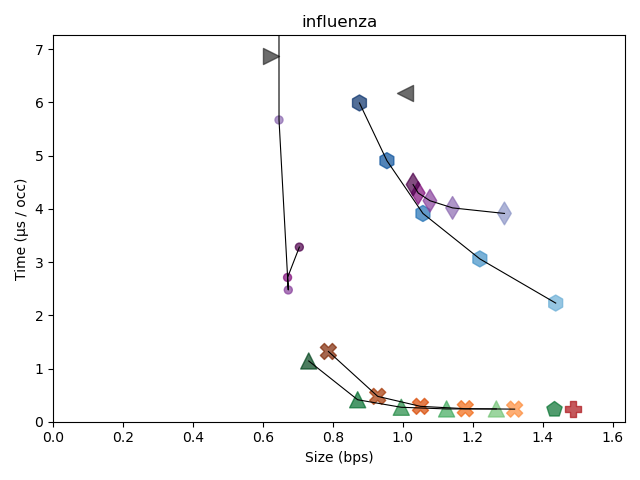}
    \includegraphics[width=0.45\textwidth]{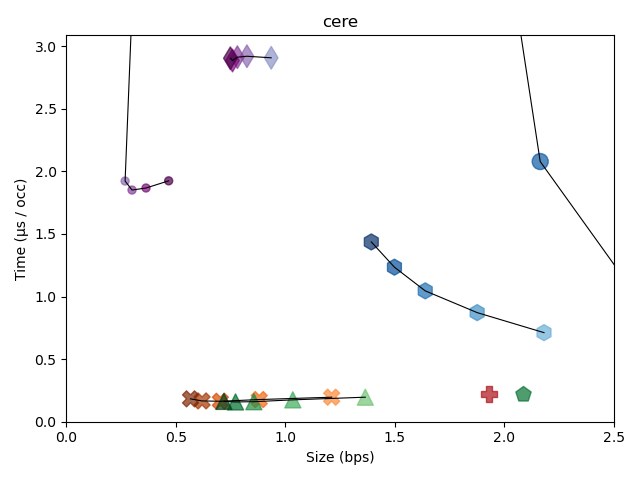}
    \includegraphics[width=0.45\textwidth]{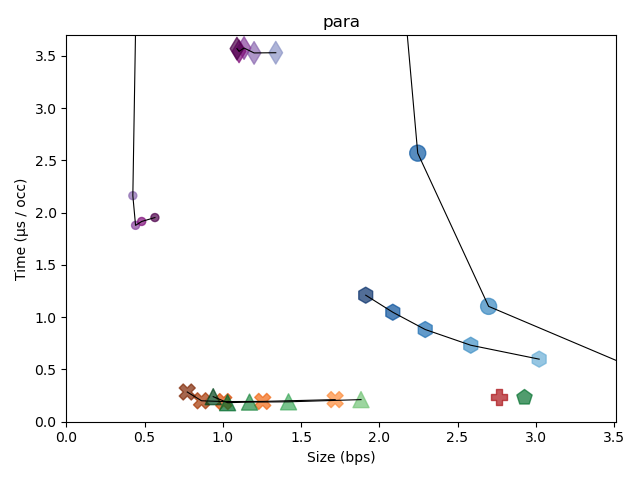}
    \includegraphics[width=0.45\textwidth]{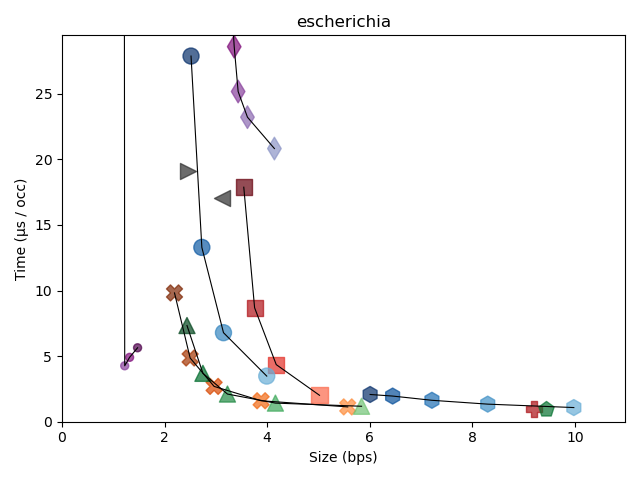}

    \includegraphics[width=0.75\textwidth]{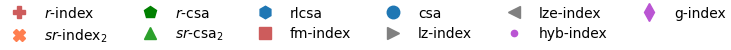}
  \end{center}
  \caption{Space-time tradeoffs on the genome PizzaChili collections.}
  \label{fig:exp-pizzachili-dna}
\end{figure}

\begin{figure}[tp]
  \begin{center}
    \includegraphics[width=0.90\textwidth]{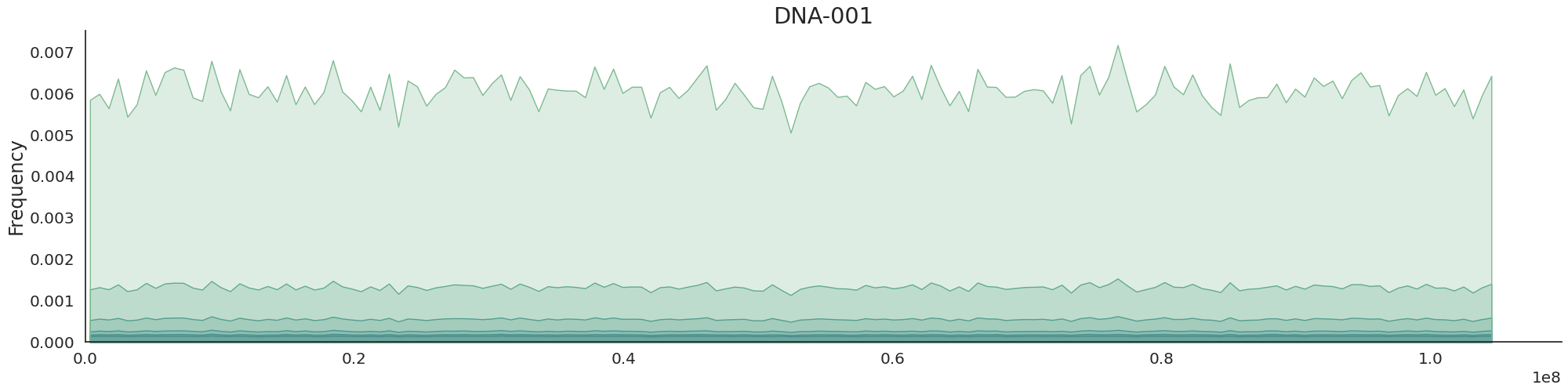}
    \includegraphics[width=0.90\textwidth]{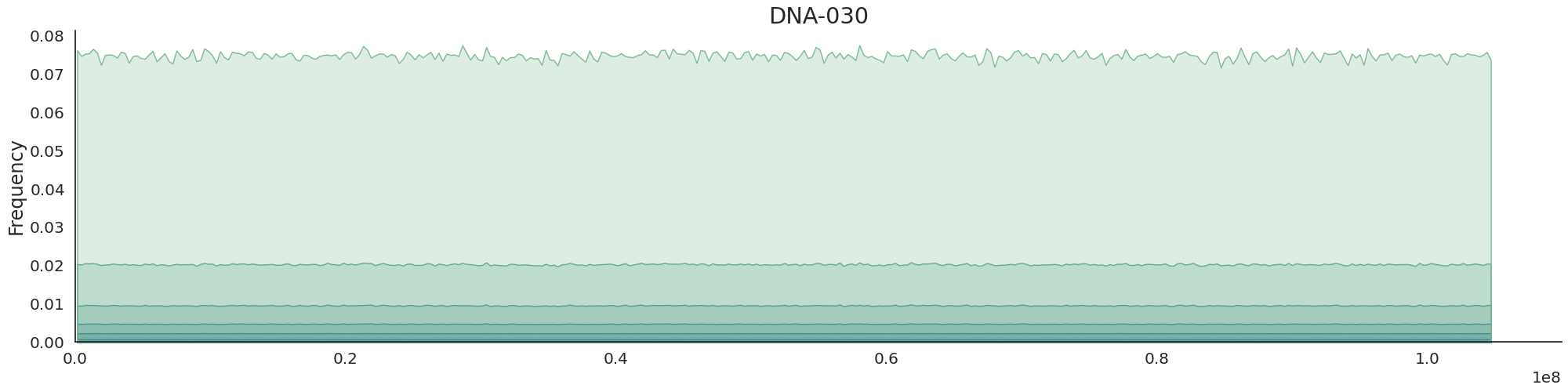}
    \includegraphics[width=0.90\textwidth]{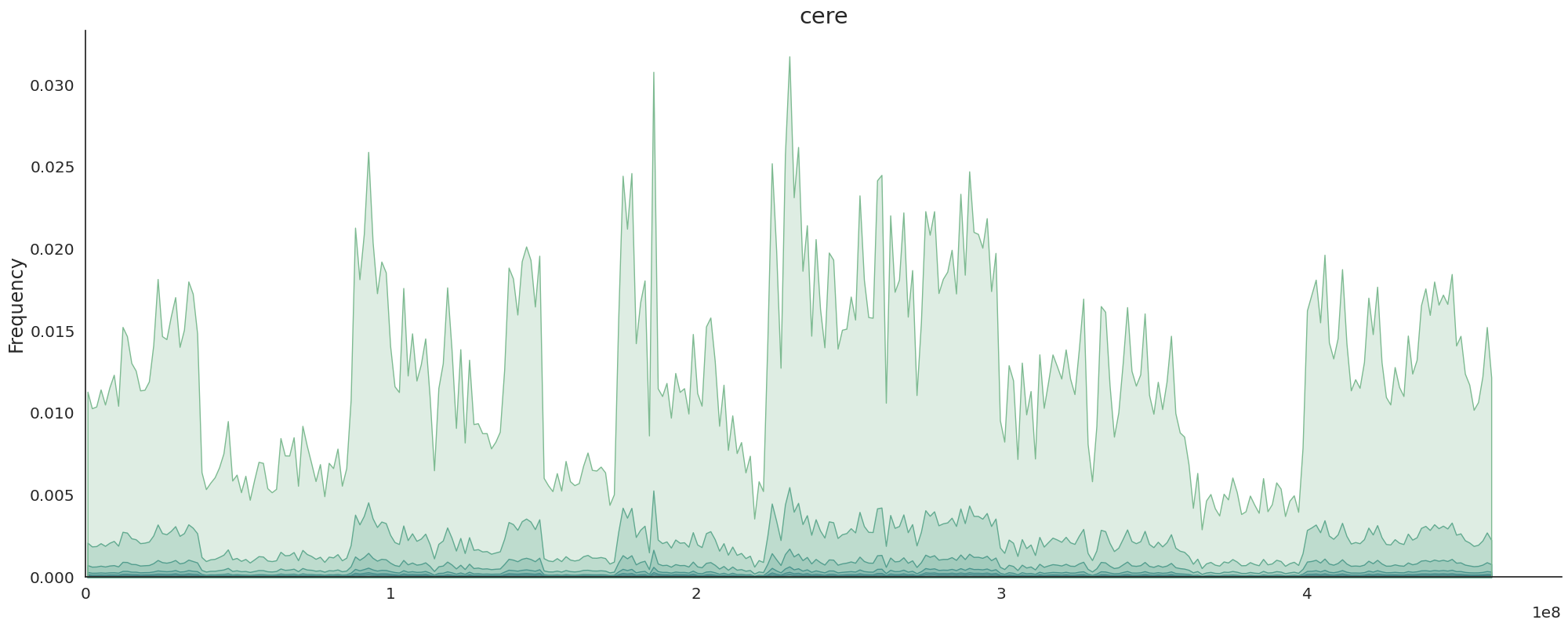}
    \includegraphics[width=0.90\textwidth]{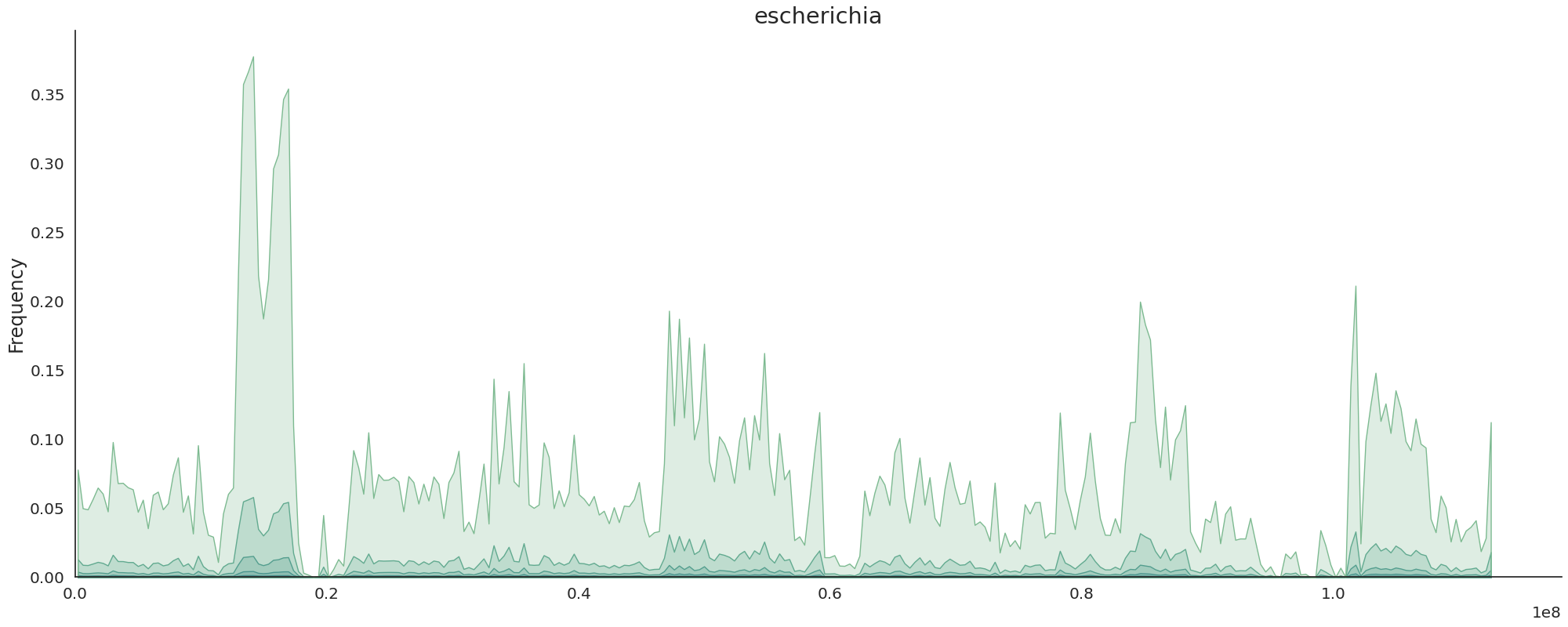}
  \end{center}
  \caption{
    Distribution of original and subsampled $\BWT$-run heads within the synthetic DNA datasets and Pizza-Chili repetitive texts.
    The $x$-axis represents the positions of these run heads along the text. The $y$-axis indicates how frequently $\BWT$-run heads appear at different locations. The original distributions are shown with the lightest color, and darker colors are used for subsampling with larger values of $s$.
  }
  \label{fig:sampling-frequencies}
\end{figure}

Figure~\ref{fig:exp-pizzachili-dna} 
shows the performance on the real-life genome collections of \nameColl{PizzaChili}. The situation is now more varied, for example the \emph{sr}-indexes now retain the performance of their corresponding \emph{r}-indexes while using $1.5$--$4.0$ times their space. The $\SRIdx$ outperforms the $\SRCSA$ in some texts and loses to it on others, independently of the average run length of the collections. Most interestingly, our \emph{sr}-indexes perform even {\em better} than their corresponding \emph{r}-indexes and \emph{rl}-indexes in those real texts, arguably because they exploit better the non-uniform distributions of the samples in the text. For example, none of our preceding indexes based on the BWT or $\Psi$ matches our \emph{sr}-indexes even on \nameColl{Escherichia}, where the average run length is around 7.5, while the $\CSA$ outperformed our subsampled indexes on the \textsf{synthetic DNA} collections with mutation probability 0.03 and average run length around 11. The $\HybridIdx$ still belongs to the Pareto curve and also seems to benefit from non-uniformity: it now clearly outperforms our \emph{sr}-indexes on \nameColl{Escherichia}, the text with the least repetitiveness.

Figure~\ref{fig:sampling-frequencies} illustrates this phenomenon. It shows the distribution of run heads in text order on some synthetic and real texts, and how many samples survive for increasing values of $s$. Note how the distribution of samples on the synthetic texts (even if coming from random mutations over an actual DNA text) are uniform and very different from the distributions on real texts. Note also how, as $s$ grows, the samples decrease much faster on the denser areas and make the distributions tend to uniform (see the darkest areas on the real texts). For example, for $s=4$, the number of samples roughly reduce to $1/4$ on the synthetic texts, while they reduce to about a $1/7$ in the densest areas of the real texts. Our worst-case analysis better reflects the uniform case, but the subsampling is much more effective on the non-uniform histograms.

\begin{figure}[t]
  \begin{center}
    \includegraphics[width=0.45\textwidth]{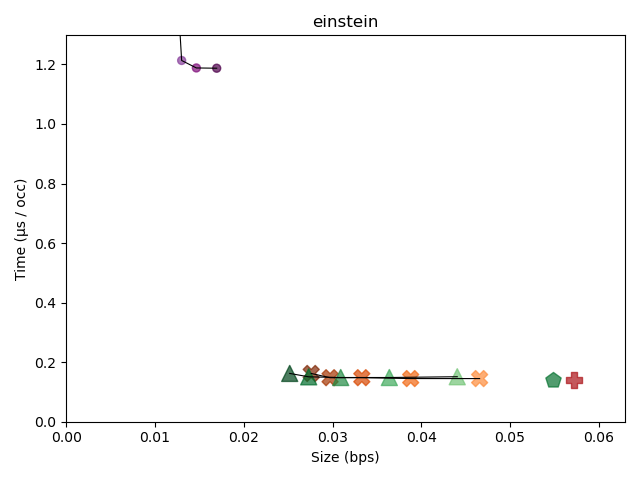}
    \includegraphics[width=0.45\textwidth]{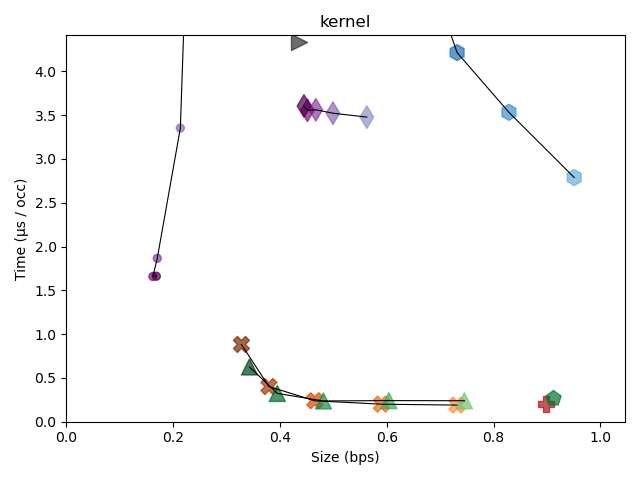}
    \includegraphics[width=0.45\textwidth]{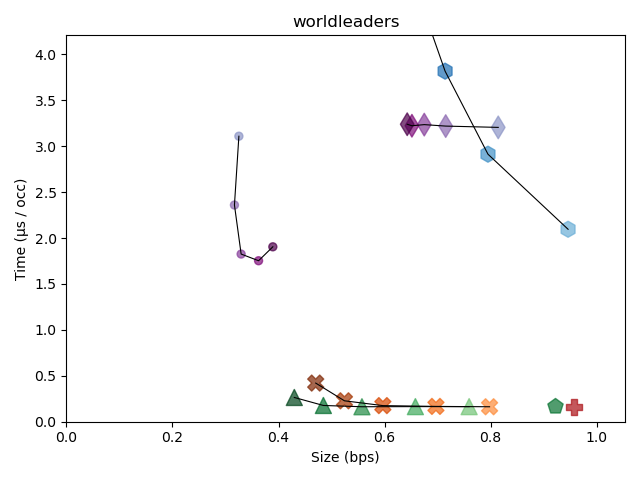}
    \includegraphics[width=0.45\textwidth]{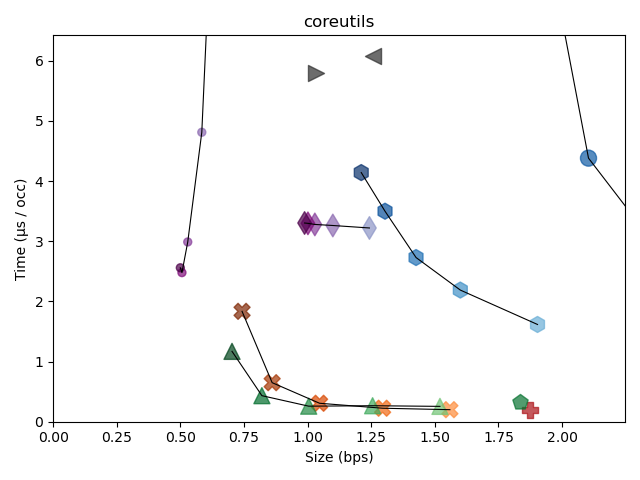}

    \includegraphics[width=0.60\textwidth]{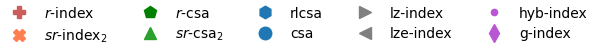}
  \end{center}
  \caption{Space-time tradeoffs on the document PizzaChili collections.}
  \label{fig:exp-pizzachili-doc}
\end{figure}

Figure~\ref{fig:exp-pizzachili-doc} shows the case of other repetitive collections in \nameColl{PizzaChili}. In these collections with larger alphabets, the $\SRCSA$ generally outperforms the $\SRIdx$, as the latter has an $\Oh(\log\sigma)$ time penalty its operations. Otherwise, the conclusions do not differ from those obtained on DNA.

Overall, we conclude that our \emph{sr}-indexes are sharply dominant when the average run length $n/r$ exceeds 10. At this point, depending on the type of text, they may be matched by other indexes, particularly the $\CSA$ and the $\HybridIdx$. Texts with those average run lengths are arguably non-repetitive anymore: even the classic indexes exceed the 2 bits per symbol used by a plain representation of the data! Finally, our \emph{sr}-indexes perform better on real than on synthetic data.


In general, the bits per symbol used by the \emph{sr}-indexes can be roughly predicted from $\vsTColl / \vnBWTRuns$; for example the sweet spot often uses around $40 \vnBWTRuns$ total bits, although it takes $20 \vnBWTRuns$--$30 \vnBWTRuns$ bits in some cases.
The $\RIdx$ and $\RCSA$ use $70 \vnBWTRuns$--$90 \vnBWTRuns$ bits.

\begin{figure}[t]
  \begin{center}
    \includegraphics[width=0.45\textwidth]{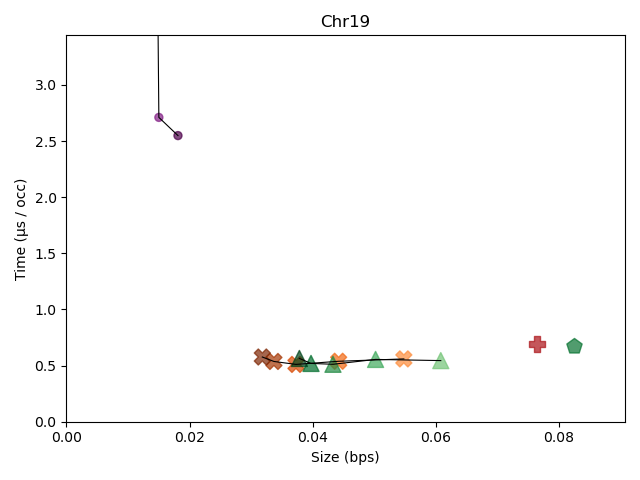}
    \includegraphics[width=0.45\textwidth]{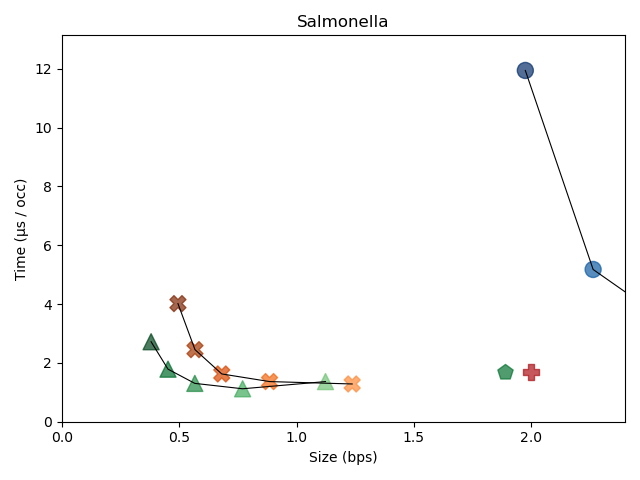}

    \includegraphics[width=0.70\textwidth]{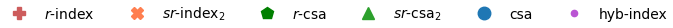}
  \end{center}
  \caption{Space-time tradeoffs on the real DNA collections.}
  \label{fig:exp-dna-real}
\end{figure}

Finally, Figure~\ref{fig:exp-dna-real} shows the results on the largest collections, \nameColl{Chr19} ($56$ GB) and \nameColl{Salmonella} ($70$ GB). We were only able to construct the $\BWT / \Psi$-related indices (except $\RLCSA$) for both, $\HybridIdx$ for \nameColl{Chr19}, and $\LZIdx$ and $\LZEndIdx$ for \nameColl{Salmonella}.
Our benchmarks show that the same observed trends scale to gigabyte-sized collections with different repetitiveness levels.
Specifically, for \nameColl{Chr19}, the $\SRIdx$ shows the best performance.
With a sampling factor of $\vSampFac = 64$, it reduces the space required by the $\RIdx$ from $0.076$ to $0.032$ bps, and even reduces the search time from $0.69$ to $0.58$ $\mu$s.
The $\HybridIdx$ requires $57\%$ of the $\RIdx$ space, but it is $4.4$ times slower.
In the case of \nameColl{Salmonella}, the $\SRCSA$ excels in the time-space tradeoff.
It reduces the $\RCSA$ space by $4.2$ times (from $1.89$ to $0.45$ bps), while the query time increases only by $4.7\%$, from $1.71$ to $1.79$ $\mu$s.


\section{Conclusions} \label{sec:conclusion}

\iftoggle{srindex}{

We have introduced the $\SRIdx$, an $\RIdx$ variant that solves the problem of its relatively bloated space while retaining its search speed. We have also designed the $\RCSA$, an equivalent to the $\RIdx$ that builds on the $\RLCSA$ instead of on the $\RLFMIdx$, and also its reduced-space version, the $\SRCSA$.
The $\SRIdx$ and $\SRCSA$ match the time performance of their large-space versions, $\RIdx$ and $\RCSA$, while using $1.5$--$4.0$ times less space. Further, they sweep the table of compressed indexes for highly repetitive text collections: they are orders of magnitude faster than the others while sharply outperforming most of them in space as well.

Unlike the $\RIdx$ and the $\RCSA$, the $\SRIdx$ and the $\SRCSA$ use little space even in scenarios of mild repetitiveness, which makes them usable in a wider range of bioinformatic applications.
For example, the $\SRIdx$ uses 0.03 bits per symbol (bps) while reporting each occurrence within half a microsecond on a very repetitive gigabyte-sized collection of human genomes, where the original $\RIdx$ and $\RCSA$ use around 0.08 bps and take the same time. On a similarly sized but less repetitive collection of Salmonella genomes, the $\SRCSA$ takes 0.45 bps to answer queries in less than 2 microseconds, matching the time of the $\RIdx$ and $\RCSA$, which take about 2 bps.

On synthetic texts with random mutations, the $\SRIdx$ and the $\SRCSA$ outperform classic compressed indexes on collections with mutation rates under as much as $1\%$. Classic indexes only match them when the mutation rate reaches 3\%. At this rate the text is arguably no longer repetitive, since no index can reach the 2 bps required to store the data in plain form. Further, our indexes perform much better on {\em real} texts than on synthetic ones, as they precisely exploit the uneven coverage of text samples that arise with the $\RIdx$ and $\RCSA$.

The conference version of this article had already an impact on Bioinformatic research.
Goga et al.~\cite{GDBFGN24} simplified and reduced the $\SRIdx$ subsample by abandoning $\fnPhi$ and $\fnIPhi$, and used that subsample in a version of Rossi et al.'s~\cite{ROLGB22} tool MONI.  Goga et al.'s version finds maximal exact matches (MEMs) between a pattern and an indexed text using longest common prefix (LCP) queries between suffixes of the pattern and suffixes starting immediately after characters at boundaries of runs in the $\BWT$ of the text.  They observed that in their experiment with 1000 human chromosome 19s, ``with s = 5 the index took less than three quarters as much space as without subsampling and used only 6\% more query time''.

The $\SRIdx$ subsampling may be useful even when we have no interest in the suffix array, as in Depuydt et al.'s~\cite{DAFGL??} index for metagenomic classification.  They use Li's~\cite{Li12} forward-backward algorithm (see a recent discussion~\cite{Li24}) to find the MEMs between a DNA long read and a large collection of genomes from several different species.  Forward-backward returns the $\BWT$ intervals for the MEMs (without using the suffix array), and Depuydt et al.\ check the corresponding intervals in a run-length compressed tag array~\cite{BGGHNPS24} indicating the species of the genome containing each character in the $\BWT$.  If there are enough sufficiently long MEMs in a read and they all occur only in genomes of the same species, Depuydt et al.\ guess the read comes from an individual of that species.

When only a few distinct species are represented in the dataset, the number of runs in the tag array may be even smaller than the number of runs in the $\BWT$.  When there are many related species, however, the $\BWT$ tends to be much more run-length compressible.  In such cases, Depuydt et al.\ store a bitvector marking the boundaries between runs in the tag array, so they can check that all the occurrences of a MEM are in genomes of one species (without finding out which species that is).  They also store as a ``toehold'' the tag for the first character in each run in the $\BWT$.  When all the occurrences of a MEM are in genomes of one species, the last toehold tag they found while computing that MEM, is that species.  Those toehold tags take the place of suffix-array entries, and they can be $\SRIdx$ subsampled as well.

Another relevant line of future work is to support direct access to arbitrary entries of the suffix array and its inverse, for example to implement compressed suffix trees \cite{GagieNP20:FFS} The $\RLFMIdx$ and $\RLCSA$ \cite{MakinenNSV10:SRH}, which for pattern searching are dominated by the $\SRIdx$ and the $\SRCSA$, use their regular text sampling to compute any such entry in time proportional to the sampling step $\vs$. Obtaining an analogous result on the $\SRIdx$ or the $\RCSA$ would lead to practical compressed suffix trees for highly repetitive text collections, which to date hardly reach the barrier of 2 bps on real bioinformatic collections \cite{NOjea15,CNic21,BCGHMNR21}.
Other proposals for accessing the suffix array faster than the $\RLFMIdx$~\cite{GNF14,PZ20} illustrate this difficulty: they require even more space than the $\RIdx$.

}{
\todo[inline]{Conclusion}
}

\section*{Acknowledgements}

Funded in part by Basal Funds FB0001, ANID, Chile.
D.C. also funded by ANID/Scholarship Program/DOCTORADO BECAS CHILE/2020-21200906, Chile.
T.G. funded by NSERC Discovery Grant RGPIN-07185-2020.
G.N. also funded by Fondecyt Grant 1-230755, ANID, Chile.



\typeout{}
\bibliographystyle{ACM-Reference-Format}
\bibliography{bibliography}


\end{document}